\documentclass[12pt]{article} 
\usepackage{algorithm}
\usepackage{algorithmic}
\usepackage{amsfonts}
\usepackage{amssymb}
\usepackage{amsmath}
\usepackage{amsthm}
\usepackage{color}
\usepackage[margin=1.2in]{geometry}
\usepackage{graphicx}
\usepackage{natbib}
\usepackage[hide]{simple_notes}
\usepackage{url}
\title{Gains in Power from Structured Two-Sample Tests of Means on Graphs}

\author{
Laurent Jacob\\
Department of Statistics\\
University of California, Berkeley\\
\texttt{laurent@stat.berkeley.edu}\\
\and
Pierre Neuvial\\
Department of Statistics\\
University of California, Berkeley\\
\texttt{pierre@stat.berkeley.edu} \\
\and
Sandrine Dudoit \\
Division of Biostatistics and Department of Statistics\\
University of California, Berkeley\\
\texttt{sandrine@stat.berkeley.edu}\\
}

\newcommand{\OMIT}[1]{}

\newtheorem{lemma}{Lemma}
\newtheorem{cor}{Corollary}

\newcommand{\1}{\mathbf{1}}

\newcommand{\E}{\mathcal{E}}
\newcommand{\G}{\mathcal{G}}
\newcommand{\HZ}{\mathbf{H}_0}
\newcommand{\HO}{\mathbf{H}_1}
\newcommand{\LL}{\mathcal{L}}
\newcommand{\N}{\mathcal{N}}

\newcommand{\RR}{\mathbb{R}}

\newcommand{\V}{\mathcal{V}}
\newcommand{\cF}{F_0}  
\newcommand{\ncF}{F}   

\graphicspath{ {images/} }

\begin{document}

\maketitle


\begin{abstract}
  We consider multivariate two-sample tests of means, where the
  location shift between the two populations is expected to be related
  to a known graph structure.  An important application of such tests
  is the detection of differentially expressed genes between two
  patient populations, as shifts in expression levels are expected to
  be coherent with the structure of graphs reflecting gene properties
  such as biological process, molecular function, regulation, or
  metabolism. For a fixed graph of interest, we demonstrate that
  accounting for graph structure can yield more powerful tests under
  the assumption of smooth distribution shift on the graph. We also
  investigate the identification of non-homogeneous subgraphs of a
  given large graph, which poses both computational and multiple
  testing problems. The relevance and benefits of the proposed
  approach are illustrated on synthetic data and on breast cancer gene
  expression data analyzed in context of KEGG pathways.
\end{abstract}

\section{Introduction}

\todo{Clearly distinguish mean/location shift $\delta$
  vs. distribution/covariance-adjusted mean shift $\Delta$.}

The problem of testing whether two data generating distributions are
equal has been studied extensively in the statistical and machine
learning literatures. Practical applications range from speech
recognition to fMRI and genomic data analysis. Parametric approaches
typically test for divergence between two distributions using
statistics based on a standardized difference of the two sample means,
\emph{e.g.}, Student's $t$-statistic in the univariate case or
Hotelling's $T^2$-statistic in the multivariate
case~\citep{Lehmann2005Testing}. A variety of non-parametric
rank-based tests have also been proposed. More recently,
\cite{Harchaoui2008Testing} and~\cite{Gretton2007A} devised
kernel-based statistics for homogeneity tests in a function space.

In several settings of interest, prior information on the structure of
the distribution shift is available as a graph on the
variables. Specifically, suppose we observe
$\{X_1^1,\ldots,X_{n_1}^1\}\in\RR^p$ from a first multivariate normal
distribution $\N(\mu_1,\Sigma)$ and
$\{X_1^2,\ldots,X_{n_2}^2\}\in\RR^p$ from a second such distribution
$\N(\mu_2,\Sigma)$.  In cases where an undirected graph $\G =
\left(\V,\E\right)$ encoding some type of covariance information in
$\RR^p$ is given, the putative \emph{location} or \emph{mean shift}
$\delta=\mu_1-\mu_2$ may be expected to be coherent with $\G$. That
is, $\delta$ viewed as a function of $\G$ is \emph{smooth}, in the
sense that the shifts $\delta_i$ and $\delta_j$ for two connected
nodes $v_i$ and $v_j\in\V$ are similar. Classical tests, such as
Hotelling's $T^2$-test, consider the null hypothesis
$\HZ~:\mu_1=\mu_2$ against the alternative $\HO~:\mu_1\neq\mu_2$,
without reference to the graph.  Our goal is to take into account the
graph structure of the variables in order to build a more powerful
two-sample test of means under smooth-shift alternatives.

Just as a natural notion of smoothness of functions on a Euclidean
space can be defined through the notion of Dirichlet energy and
controlled by Fourier decomposition and
filtering~\citep{Stain1971Introduction}, it is
well-known~\citep{Chung1997Spectral} that the smoothness of functions
on a graph can be naturally defined and controlled through spectral
analysis of the graph Laplacian.  In particular, the eigenvectors of
the Laplacian provide a basis of functions which vary on the graph at
increasing frequencies (corresponding to the increasing
eigenvalues). In this paper, we propose to compare two populations in
terms of the first few components of the graph-Fourier basis or,
equivalently, in the original space, after filtering out
high-frequency components.


An important motivation for the development of our graph-structured
test is the detection of groups of genes whose expression changes
between two conditions. For example, identifying groups of genes that
are differentially expressed (DE) between patients for which a
particular treatment is effective and patients which are resistant to
the treatment may give insight into the resistance mechanism and even
suggest targets for new drugs.  In such a context, expression data
from high-throughput microarray and sequencing assays gain much in
relevance from their association with graph-structured prior
information on the genes, \emph{e.g.}, Gene Ontology (GO;
\url{http://www.geneontology.org}) or Kyoto Encyclopedia of Genes and
Genomes (KEGG; \url{http://www.genome.jp/kegg}).  Most approaches to
the joint analysis of gene expression data and gene graph data involve
two distinct steps. Firstly, tests of differential expression are
performed separately for each gene. Then, these univariate
(gene-level) testing results are extended to the level of gene sets,
\emph{e.g.}, by assessing the over-representation of DE genes in each
set based on $p$-values for Fisher's exact test (or a $\chi^2$
approximation thereof) adjusted for multiple
testing~\citep{Beissbarth2004GOstat} or based on permutation adjusted
$p$-values for weighted Kolmogorov-Smirnov-like
statistics~\citep{Subramanian2005Gene}. Another family of methods
directly performs multivariate tests of differential expression for
groups of genes, \emph{e.g.}, Hotelling's
$T^2$-test~\citep{Lu2005Hotelling}. It is
known~\citep{Goeman2007Analyzing} that the former family of approaches
can lead to incorrect interpretations, as the sampling units for the
tests in the second step become the genes (as opposed to the patients)
and these are expected to have strongly correlated expression
measures. This suggests that direct multivariate testing of gene set
differential expression is more appropriate than posterior aggregation
of individual gene-level tests. On the other hand, while Hotelling's
$T^2$-statistic is known to perform well in small dimensions, it loses
power very quickly with increasing dimension~\citep{Bai1996Effect},
essentially because it is based on the inverse of the empirical
covariance matrix which becomes ill-conditioned. In addition, such
direct multivariate tests on unstructured gene sets do not take
advantage of information on gene regulation or other relevant
biological properties.  An increasing number of regulation networks
are becoming available, specifying, for example, which genes activate
or inhibit the expression of which other genes.  As stated before,
incorporating such biological knowledge in DE tests is
important. Indeed, if it is known that a particular gene in a tested
gene set activates the expression of another, then one expects the two
genes to have coherent (differential) expression patterns,
\emph{e.g.}, higher expression of the first gene in resistant patients
should be accompanied by higher expression of the second gene in these
patients.  Accordingly, the first main contribution of this paper is
to propose and validate multivariate test statistics for identifying
distribution shifts that are coherent with a given graph structure.

Next, given a large graph and observations from two data generating
distributions on the graph, a more general problem is the
identification of smaller non-homogeneous subgraphs, \emph{i.e.},
subgraphs on which the two distributions (restricted to these
subgraphs) are significantly different. This is very relevant in the
context of tests for gene set differential expression: given a large
set of genes, together with their known regulation network, or the
concatenation of several such overlapping sets, it is important to
discover novel gene sets whose expression change significantly between
two conditions.  Currently-available gene sets have often been defined
in terms of other phenomena than that under study and physicians may
be interested in discovering sets of genes affecting in a concerted
manner a specific phenotype.  Our second main contribution is
therefore to develop algorithms that allow the exhaustive testing of
all the subgraphs of a large graph, while accounting for the
multiplicity issue arising from the vast number of subgraphs.


As the problem of identifying variables or groups of variables which
differ in distribution between two populations is closely-related to
supervised learning, our proposed approach is similar to several
learning methods. \cite{Rapaport2007Classification} use filtering in
the Fourier space of a graph to train linear classifiers of gene
expression profiles whose weights are smooth on a gene
network. However, their classifier enforces global smoothness on the
large regularization network of all the genes, whereas we are
concerned with the selection of gene sets with locally-smooth
expression shift between populations. In~\cite{Jacob2009Group}, sparse
learning methods are used to build a classifier based on a small
number of gene sets. While this approach leads in practice to the
selection of groups of variables whose distributions differ between
the two classes, the objective is to achieve the best classification
performance with the smallest possible number of groups. As a result,
correlated groups of variables are typically not selected. Other
related work includes~\cite{Fan1998Test}, who proposed an adaptive
Neyman test in the Fourier space for time-series. However, as
illustrated below in Section~\ref{sec:experiments}, direct translation
of the adaptive Neyman statistic to the graph case is problematic, as
assumptions on Fourier coefficients which are true for time-series do
not hold for graphs.  In addition, the Neyman statistic converges very
slowly towards its asymptotic distribution and the required
calibration by bootstrapping renders its application to our subgraph
discovery context difficult. By contrast, other methods do not account
for shift smoothness and try to address the loss of power caused by
the poor conditioning of the $T^2$-statistic by applying it after
dimensionality reduction~\citep{Ma2009Identification} or by omitting
the inverse covariance matrix and adjusting instead by its
trace~\citep{Bai1996Effect,Chen2010A}. \cite{Vaske2010Inference}
recently proposed DE tests, where a probabilistic graphical model is
built from a gene network. However, this model is used for gene-level
DE tests, which then have to be combined to test at the level of gene
sets. Several approaches for subgraph discovery, like that
of~\cite{Ideker2002Discovering}, are based on a heuristic to identify
the most differentially expressed subgraphs and do not amount to
testing exactly all the subgraphs. Concerning the discovery of
distribution-shifted subgraphs, \cite{Vandin2010Algorithms} propose a
graph Laplacian-based testing procedure to identify groups of
interacting proteins whose genes contain a large number of
mutations. Their approach does not enforce any smoothness on the
detected patterns (smoothness is not necessarily expected in this
context) and the graph Laplacian is only used to ensure that very
connected genes do not lead to spurious detection. The Gene Expression
Network Analysis (GXNA) method of \cite{Nacuetal07} detects
differentially expressed subgraphs based on a greedy search algorithm
and gene set DE scoring functions that do not account for the graph
structure.


The rest of this paper is organized as follows:
Section~\ref{sec:smooth} introduces elements of Fourier analysis for
graphs which are needed to develop our method. Section~\ref{sec:test}
presents our graph-structured two-sample test statistic and states
results on power gain for smooth-shift
alternatives. Section~\ref{sec:discovery} describes procedures for
systematically testing all the subgraphs of a large graph.
Section~\ref{sec:experiments} presents results for synthetic data as
well as breast cancer gene expression and KEGG data.  Finally,
Section~\ref{sec:discussion} summarizes our findings and outlines
ongoing work.


\section{Fourier analysis on graphs}
\label{sec:smooth}

The fundamental idea of harmonic analysis for functions defined on a
Euclidean space is to build a basis of the function space, such that
each basis function varies at a different frequency. The basis
functions are typically sinusoids. They were originally obtained in an
attempt to solve the heat equation, as the eigenfunctions of the
Laplace operator, with corresponding eigenvalues proportional to the
frequencies of the sinusoids. Any function can then be decomposed on
the basis as a linear combination of sinusoids of increasing
frequency. The set of projections of the function on the basis
sinusoids gives a dual representation of the function, often referred
to as Fourier transform. This representation is useful for filtering
functions, by removing or shrinking coefficients associated with high
frequencies, as these are typically expected to reflect noise, and
then taking the inverse Fourier transform.  The resulting filtered
function contains the same signal in the low frequencies as the
original function. A related concept is the Dirichlet energy of a
function $f$ on an open subspace $\Omega$, defined as
$\frac{1}{2}\int_{\Omega}|\nabla f(x)|^2dx$ where $\nabla$ is the
gradient operator, a measure of variation that is consistent with the
Laplace operator.  In particular, the Dirichlet energy of the basis
functions is proportional to their associated frequencies.

\todo{LJ: I think it can directly be written in terms of the
  laplacian. Defn $d_i$, $f_i$. Defn fct on graph.}

For functions on a Euclidean space, natural notions of smoothness,
along with the Dirichlet energy and dual representation in the
frequency domain by projection on a Fourier basis, are therefore
classically defined from the Laplace operator and its spectral
decomposition.  Likewise, notions of smoothness for functions on
graphs can be defined based on the graph Laplacian.  Specifically,
consider an undirected graph $\G = \left(\V,\E\right)$, with $|\V| =
p$ nodes, adjacency matrix $A$, and degree matrix $D =
\text{Diag}\left(A \1\right)$, where $\1$ is a unit column-vector,
$\text{Diag}(x)$ is the diagonal matrix with diagonal $x$ for any
vector $x$, and $D_{ii} = d_i$. Let $f:\RR^{|\V|}\rightarrow \RR$
denote a function that associates a real value to each node of the
graph $\G$. The Laplacian matrix of $\G$ is typically defined as $\LL
= D - A$ or $\LL_{norm} = I - D^{-\frac{1}{2}}AD^{-\frac{1}{2}}$ for
the normalized version. More generally, given any gradient matrix
$\nabla \in \RR^{|\E|,|\V|}$, defined on $\G$ and associating to each
function on the graph its variation on each edge, it is possible to
derive a corresponding Laplacian matrix following the classical
definition of the Laplace operator, $\LL = -\textrm{div} \nabla =
\nabla^\top \nabla$, where $\textrm{div}$ is the divergence operator
defined as the negative of the adjoint operator of the
gradient~\citep{Zhou2005Regularization}.  Any desired notion of
variation may be encoded in a gradient function and thus translated
into its associated Dirichlet energy $\frac{1}{2}f^\top\LL f$, for a
function $f$ defined on the graph $\G$. A common choice of gradient is
the finite difference operator $\nabla f = \left(f_i -
  f_j\right)_{i,j\in \V}$. This definition leads to the unnormalized
Laplacian above. The corresponding energy function is
$\frac{1}{2}\sum_{i,j\in\V}\left(f_i - f_j\right)^2$.  Let
$\LL=U\Lambda U^\top$ denote the spectral decomposition of the
Laplacian, where $\Lambda$ is the diagonal matrix of eigenvalues
$\lambda_i$ and the columns of the matrix $U$ are the corresponding
eigenvectors $u_i$. Then, by definition, the eigenvectors of $\LL$ are
functions of increasing energy, as $u_i^\top\LL u_i = \lambda_i$ for
all $i=1,\ldots,p$.  In the remainder of this paper, we denote by
$\tilde{f} = U^\top f$ the Fourier coefficients of a function $f$
defined on a graph.

If the above two notions of smoothness are not appropriate for a
particular application, other gradients, leading to other Laplacian
matrices, may be devised to build the function basis. For example,
introducing weights on the edges of a graph and using these weights in
the normalized version of the finite differences allows the
incorporation of prior belief on where a shift in distributions is
expected to be smooth. For applications like structured gene set
differential expression detection, one may use negative weights for
edges that reflect an expected negative correlation between two
variables, \emph{e.g.}, a gene $i$ whose expression inhibits the
expression of another gene $j$. In this case, a small variation of the
shift on the edge between $i$ and $j$ should correspond to a small
$|\delta_i + \delta_j|$. Accordingly, the gradient should be defined
as $\left(f_i - s_{ij} f_j\right)_{i,j\in\V}$, where $s_{ij}$ is $-1$
for negative interactions and $1$ for positive interactions. The
eigenvectors of the corresponding Laplacian $\LL_{\text{sign}}$ are
functions of increasing $\frac{1}{2}\sum_{i,j\in\V}\left(f_i - s_{ij}
  f_j\right)^2$, an appropriate notion of smoothness for the
application at hand. A signed Laplacian can be recovered from the
classical definition $\LL_{\text{sign}} = D - A_{\text{sign}}$, where
$A_{\text{sign}}$ is allowed to have negative entries. Note that such
a smoothness function is used as a penalty for semi-supervised
learning in~\cite{Goldberg2007Dissimilarity}.

As an example, Figure~\ref{fig:ev} displays the eigenvectors of the
signed Laplacian $\LL_{\text{sign}}$ for a simple four-node graph with
$$
D = 
\left (
\begin{array}{cccc}
1 & 0 & 0 & 0\\
0 & 3 & 0 & 0\\
0 & 0 & 1 & 0\\
0 & 0 & 0 & 1
\end{array}
\right ), \quad
A_{\text{sign}} = 
\left (
\begin{array}{cccc}
0 & 1 & 0 & 0\\
1 & 0 & 1 & -1\\
0 & 1 & 0 & 0\\
0 & -1 & 0 & 0
\end{array}
\right ), \quad
\LL_{\text{sign}} = 
\left (
\begin{array}{cccc}
1 & -1 & 0 & 0\\
-1 & 3 & -1 & 1\\
0 & -1 & 1 & 0\\
0 & 1 & 0 & 1
\end{array}
\right ).
$$
The first eigenvector, corresponding to the smallest frequency
(eigenvalue of zero), can be viewed as a ``constant'' function on the
graph, in the sense that its absolute value is identical for all the
nodes, but nodes connected by an edge with negative weight take on
values of opposite in sign.  By contrast, the last eigenvector,
corresponding to the highest frequency, is such that nodes connected
by positive edges take on values of opposite sign and nodes connected
by negative edges take on values of the same sign.

\begin{figure}[h]
  \begin{center}
    \includegraphics[width=0.24\linewidth]{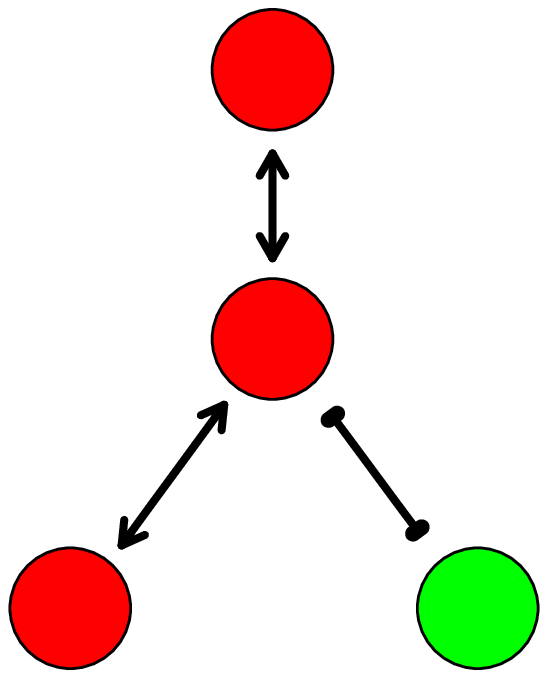}
    \includegraphics[width=0.24\linewidth]{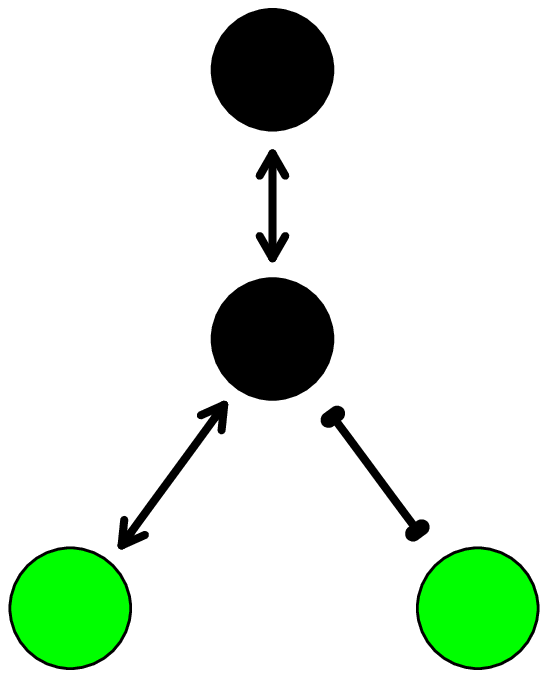}
    \includegraphics[width=0.24\linewidth]{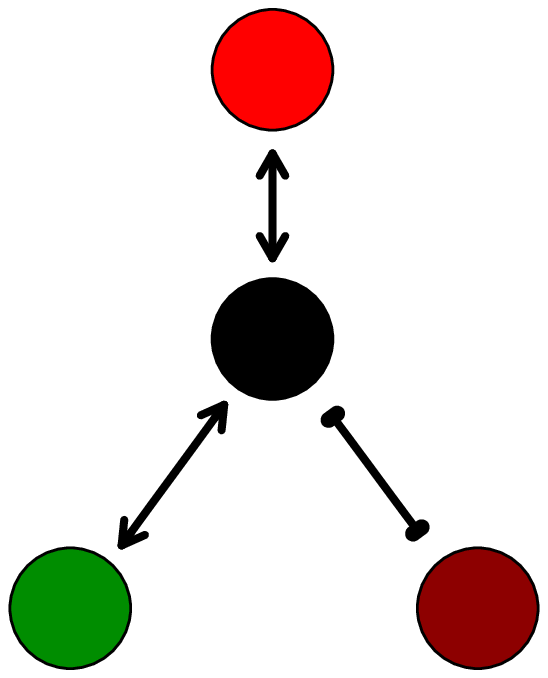}
    \includegraphics[width=0.24\linewidth]{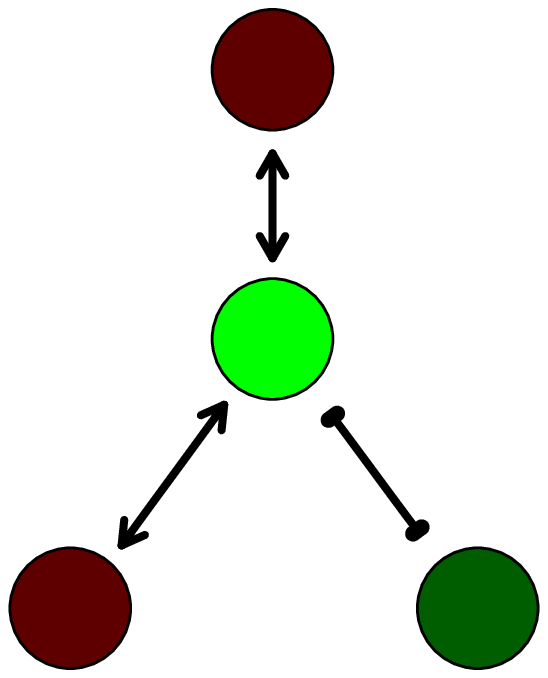}
  \end{center}
  \caption{Eigenvectors of the signed Laplacian $\LL_{\text{sign}}$
    for a simple four-node graph. The corresponding eigenvalues are
    $0, 1, 1, 4$. Nodes are colored according to the value of the
    eigenvector, where green corresponds to high positive values, red
    to high negative values, and black to $0$. ``T''-shaped edges have
    negative weights.}
  \label{fig:ev}
\end{figure}
 \todo{Figure~\ref{fig:ev}: negative interactions are hard to see.  Use a struck (i.e. barr\'ee) edge ?}

\section{Graph-structured two-sample test of means under smooth-shift alternatives}
\label{sec:test}

For multivariate normal distributions, Hotelling's $T^2$-test, a
classical test of location shift, is known to be uniformly most
powerful invariant against global-shift alternatives. The test
statistic is based on the squared \emph{Mahalanobis norm} of the
sample mean shift and is given by $T^2 =
\frac{n_1n_2}{n_1+n_2}(\bar{x}_1 - \bar{x}_2)^\top \hat{\Sigma}^{-1}
(\bar{x}_1 - \bar{x}_2)$, where $n_i$, $\bar{x}_i$, and
$\hat{\Sigma}^{-1}$ denote, respectively, the sample sizes, means, and
pooled covariance matrix, for random samples drawn from two
$p$-dimensional Gaussian distributions, $\N(\mu_i,\Sigma)$,
$i=1,2$. Under the null hypothesis $\HZ~:\mu_1=\mu_2$ of equal means,
$N T^2$ follows a (central) $F$-distribution $\cF(p,n_1+n_2-p-1)$, where
$N = \frac{n_1+n_2-p-1}{(n_1+n_2-2)p}$. In general, $N T^2$ follows a
non-central $F$-distribution $\ncF(\frac{n_1
  n_2}{n_1+n_2}\Delta^2(\delta,\Sigma); p,n_1+n_2-p-1)$, where the
non-centrality parameter is a function of the Mahalanobis norm of the
mean shift $\delta$, $\Delta^2(\delta,\Sigma) = \delta^\top
\Sigma^{-1}\delta$, which we refer to as \emph{distribution shift}. In
the remainder of this paper, unless otherwise specified,
$T^2$-statistics are assumed to follow the nominal $F$-distribution,
\emph{e.g.}, for critical value and power calculations.

For any graph-Fourier basis $U$, direct calculation shows that $T^2 =
\tilde{T}^2 \stackrel{\Delta}{=} \frac{n_1n_2}{n_1+n_2}(\bar{x}_1 -
\bar{x}_2)^\top U \left(U^\top \hat{\Sigma} U\right)^{-1}U^\top
(\bar{x}_1 - \bar{x}_2)$, \emph{i.e.}, the statistic $T^2$ in the
original space and the statistic $\tilde{T}^2$ in the graph-Fourier
space are identical. More generally, for $k \leq p$, the statistic in
the original space after filtering out frequencies above $k$ is the
same as the statistic $\tilde{T}_k^2$ restricted to the first $k$
coefficients in the graph-Fourier space:
\begin{align*}
\tilde{T}_k^2
&\stackrel{\Delta}{=} \frac{n_1n_2}{n_1+n_2}(\bar{x}_1 -
\bar{x}_2)^\top U_{[k]} \left(U_{[k]}^\top \hat{\Sigma}
U_{[k]}\right)^{-1}U_{[k]}^\top (\bar{x}_1 - \bar{x}_2) \\
&= \frac{n_1n_2}{n_1+n_2}(\bar{x}_1 - \bar{x}_2)^\top U 1_{k} U^\top
\left(U 1_{k} U^\top \hat{\Sigma}
  U 1_{k} U^\top\right)^{+}U 1_{k} U^\top (\bar{x}_1 -
\bar{x}_2),
\end{align*}
where $A^{+}$ denotes the generalized inverse of a matrix $A$, the $p \times k$ matrix $U_{[k]}$ denotes the restriction of $U$
to its first $k$ columns, and $1_{k}$ is a $p \times p$ diagonal
matrix, with $i$th diagonal element equal to one if $i \leq k$ and
zero otherwise.  Note that as retaining the first $k$ Fourier
components is a \emph{non-invertible} transformation, this filtering
indeed has an effect on the test statistic, that is, we have
$\tilde{T}_k^2 \neq \tilde{T}^2$ in general.  As the Mahalanobis norm
is invariant to linear invertible transformations, using an invertible
filtering (such as weighting each Fourier component according to its
corresponding eigenvalue) would have no impact on the test statistic.

Hotelling's $T^2$-test is known to behave poorly in high dimension;
the following lemma shows that gains in power can be achieved by
filtering. Specifically, let $\tilde{\delta} = U^{\top}\delta$ and
$\tilde{\Sigma} = U^\top\Sigma U$ denote, respectively, the mean shift
and covariance matrix in the graph-Fourier space.  Given $k \leq p$,
let $\Delta_{k}^2(\delta,\Sigma) =
\delta_{[k]}^\top\left(\Sigma_{[k]}\right)^{-1}\delta_{[k]}$ denote
the distribution shift restricted to the first $k$ dimensions of
$\delta$ and $\Sigma$, \emph{i.e.}, based on only the first $k$
elements of $\delta$, $(\delta_i: i\leq k)$, and the first $k \times
k$ diagonal block of $\Sigma$, $(\sigma_{ij}: i, j \leq k)$. Under the
assumption that the distribution shift is smooth, \emph{i.e.}, lies
mostly at the beginning of the graph spectrum, so that
$\Delta_{k}^2(\tilde{\delta},\tilde{\Sigma})$ is nearly maximal for a
small value of $k$, Lemma~\ref{lem:das} states that performing
Hotelling's test in the graph-Fourier space restricted to its first
$k$ components yields more power than testing in the full
graph-Fourier space. Equivalently, the test is more powerful in the
original space after filtering than in the original unfiltered space.
Note that this result holds because retaining the first $k$ Fourier
components is a \emph{non-invertible} transformation.

\begin{lemma}
  \label{lem:das} For any level $\alpha$ and any $1<l\leq p-k$, there
  exists $d(\alpha,k,l) > 0$ such that
  \begin{displaymath}
   \Delta_{k+l}^2(\tilde{\delta},\tilde{\Sigma}) - \Delta_k^2(\tilde{\delta},\tilde{\Sigma}) <
   d(\alpha,k,l) \Rightarrow \beta_{\alpha,k}(\Delta_k^2(\tilde{\delta},\tilde{\Sigma})) >
   \beta_{\alpha,k+l}(\Delta_{k+l}^2(\tilde{\delta},\tilde{\Sigma})),
  \end{displaymath}
  where $\beta_{\alpha,k}(\Delta^2)$ is the power of Hotelling's
  $T^2$-test at level $\alpha$ in dimension $k$ for a distribution
  shift $\Delta^2$, according to the nominal $F$-distribution $\ncF(\frac{n_1
  n_2}{n_1+n_2}\Delta^2; k,n_1+n_2-k-1)$.
\end{lemma}
\begin{proof}
  This lemma is a direct application of Corollary $2.1$
  in~\cite{Das1974Power} to Hotelling's $T^2$-test in the
  graph-Fourier space. The bottom line of the proof of~\cite{Das1974Power}'s result is that $\beta_{\alpha,k}$ can
  be shown to be a continuous and strictly decreasing function of $k$,
  so that a strictly positive increase in the non-centrality parameter
  $\Delta^2$ of the $F$-distribution is necessary to maintain 
  power when increasing dimension.
\end{proof}

In particular, a direct application of Lemma~\ref{lem:das} yields the
following Corollary:
\begin{cor}
  \label{cor:smooth}
  If $\forall \ 1<l\leq p-k,\;\Delta_k^2(\tilde{\delta},\tilde{\Sigma}) =
  \Delta_{k+l}^2(\tilde{\delta},\tilde{\Sigma})$, then
  \begin{displaymath}
    \beta_{\alpha,k}(\Delta_k^2(\tilde{\delta},\tilde{\Sigma})) > \beta_{\alpha,k+l}(\Delta_{k+l}^2(\tilde{\delta},\tilde{\Sigma})).
  \end{displaymath}
\end{cor}

\OMIT{
  In particular, it is easy to check that if $\forall \ 1<l\leq
  p-k,\;\Delta_k^2(\tilde{\delta},\tilde{\Sigma}) =
  \Delta_{k+l}^2(\tilde{\delta},\tilde{\Sigma})$, then
  $\beta_{\alpha,k}(\Delta_k^2(\tilde{\delta},\tilde{\Sigma})) >
  \beta_{\alpha,k+l}(\Delta_{k+l}^2(\tilde{\delta},\tilde{\Sigma}))$.
  Therefore, }

According to Corollary~\ref{cor:smooth}, if the distribution shift
lies in the first $k$ Fourier coefficients, then testing in this
subspace yields strictly more power than using additional
coefficients. In particular, if there exists $k<p$ such that
$\tilde{\delta}_j = 0 \;\forall\ j>k$ (\emph{i.e.}, the mean shift is
smooth) and $\tilde{\Sigma}$ is block-diagonal such that
$\tilde{\sigma}_{ij} = 0\;\forall\; i<k, j>k$, then gains in power are
obtained by testing in the first $k$ Fourier components. Although
non-necessary, this condition is plausible when the mean shift lies at
the beginning of the spectrum, as the coefficients which do not
contain the shift are not expected to be correlated with the ones that
do contain it.

Note that the result in Lemma~\ref{lem:das} is even more general, as
testing in the first $k$ Fourier components can increase power even
when the distribution shift partially lies in the remaining
components, as long as the latter portion is below a certain
threshold. Figure~\ref{fig:shiftInc} illustrates, under different
settings, the increase in distribution shift necessary to maintain a
given power level against the number of added coefficients.

\todo{LJ: Ajouter lemma.}  

If for some reason one expects that the
mean shift $\delta$ is smooth (rather than the distribution shift
$\Delta$), \emph{i.e.}, $\tilde{\delta}$ lies at the beginning of the
spectrum, and that the covariance between coefficients that contain
the shift and those that do not is non-zero, then one should use test
statistics based on estimators of the unstandardized \emph{Euclidean
  norm} $\|\delta\|$ of this shift, \emph{e.g.},
$Z$~\citep{Bai1996Effect}[Equation $(4.5)$] or
$T_n$~\citep{Chen2010A}. Results similar to Lemma~\ref{lem:das} can be
derived for these statistics. Namely, the corresponding tests gain
asymptotic power when applied at the beginning of the spectrum,
provided the Euclidean norm of $\delta$ only increases moderately as
coefficients for higher frequencies are added. The results follow
from~\cite{Bai1996Effect}[Theorem~$4.1$] and
\cite{Chen2010A}[Equations $(3.11)$--$(3.12)$], using the fact that,
by Cauchy's interlacing theorem, the trace of the square of any
positive semi-definite matrix is larger than the trace of the square
of any principal submatrix.


\begin{figure}[h]
  \begin{center}
      \includegraphics[width=0.45\linewidth]{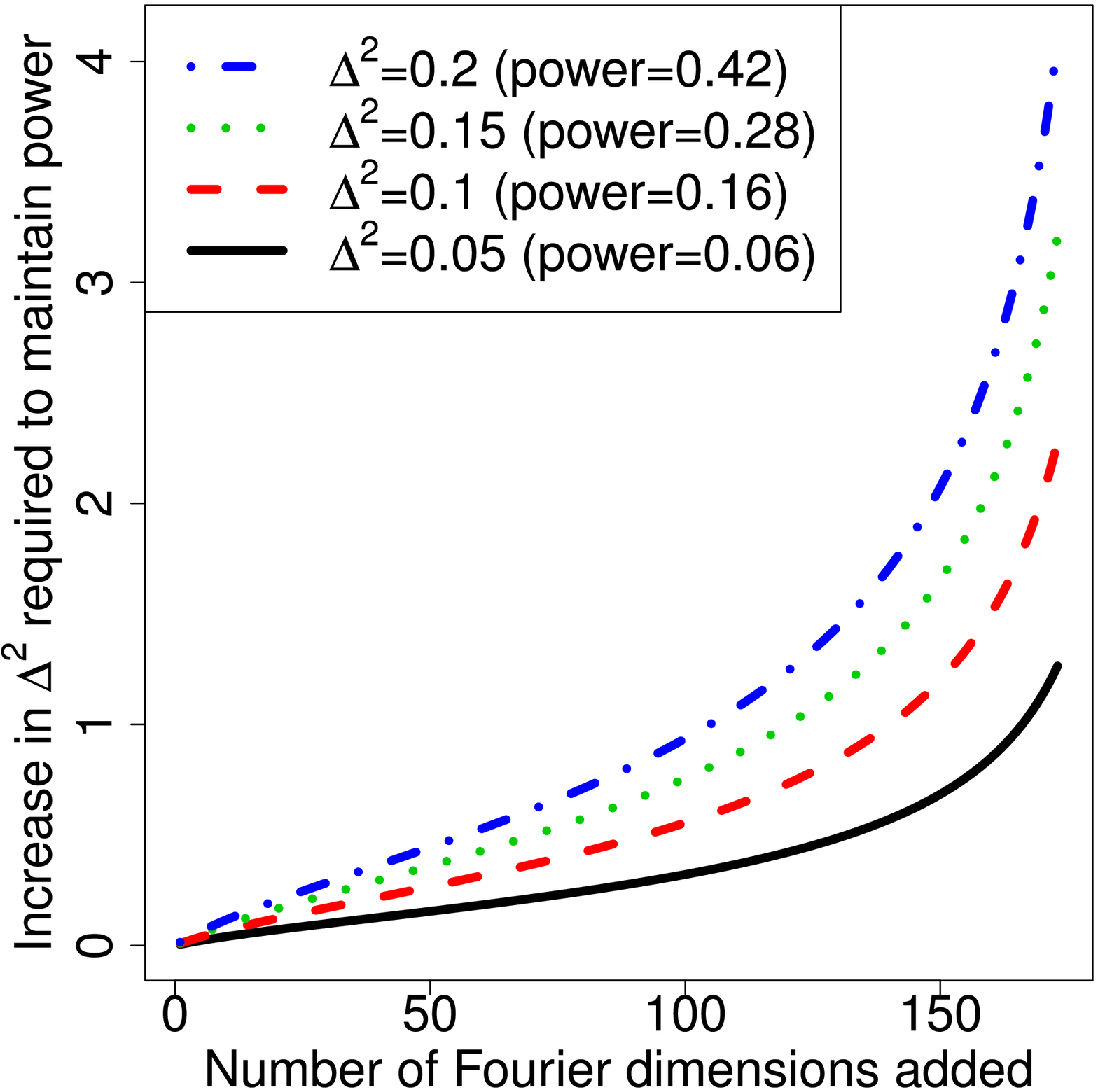}
      \hfill
      \includegraphics[width=0.45\linewidth]{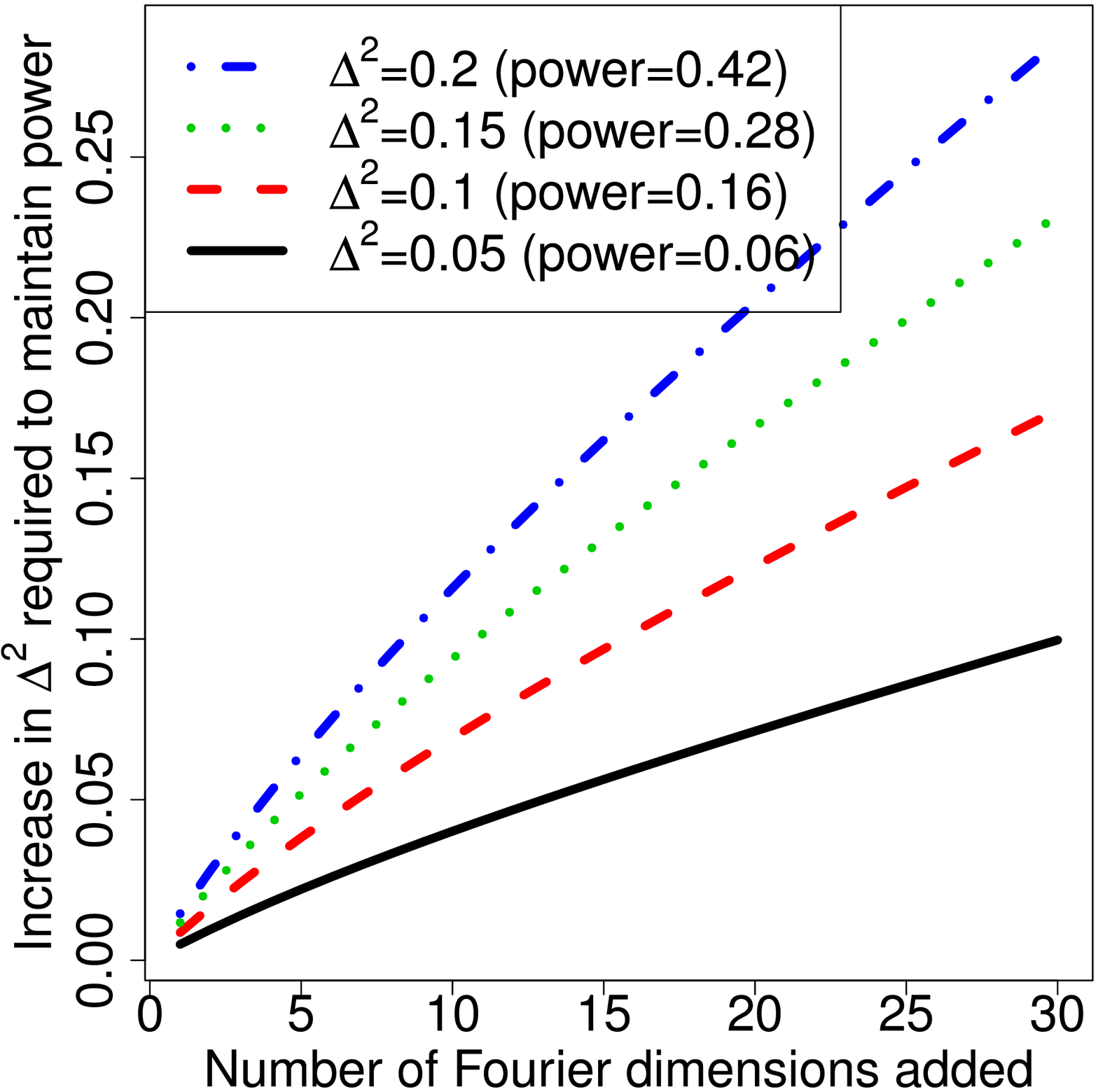}
  \end{center}
  \caption{Left: Increase in distribution shift required for
    Hotelling's $T^2$-test to maintain a given power when increasing
    the number of tested Fourier coefficients: $\Delta_{k+l}^2 -
    \Delta_{k}^2$ vs. $l$ such that
    $\beta_{\alpha,k+l}(\Delta_{k+l}^2)=\beta_{\alpha,k}(\Delta_k^2)$.
    Power $\beta_{\alpha,k+l}(\Delta_{k+l}^2)$ computed under the
    non-central $F$-distribution $\ncF(\frac{n_1
      n_2}{n_1+n_2}\Delta_{k+l}^2; k+l,n_1+n_2-(k+l)-1)$, for
    $n_1= n_2= 100$ observations, $k=5$, and $\alpha=10^{-2}$. Each line
    corresponds to the fixed shift $\Delta_{k}^2$ and power
    $\beta_{\alpha,k}(\Delta_{k}^2)$ pair indicated in the
    legend. Right: Zoom on the first $30$
    dimensions. 
    }
  \label{fig:shiftInc}
\end{figure}

\section{Non-homogeneous subgraph discovery}
\label{sec:discovery}

A systematic approach for discovering non-homogeneous subgraphs,
\emph{i.e.}, subgraphs of a large graph that exhibit a significant
shift in means, is to test all of them. In practice, however, this can
represent an intractable number of tests, so it is important to be
able to rapidly identify sets of subgraphs that all satisfy the null
hypothesis of equal means. To this end, we devise a pruning approach
based on an upper bound on the value of the test statistic for any
subgraph containing a given set of nodes.

\todo{Elaborate: combinatorial problem, exhaustive search has exponential complexity...}
\todo{Mention that this can technically be done with or without filtering.}


\todo{Simple statement of the algorithm, e.g., pseudo-code.}

\subsection{Exact algorithm}
\label{sec:exact}


Given a large graph $\G$ with $p$ nodes, we adopt the following
classical branch-and-bound-like approach to test subgraphs of size $q
\leq p$ at level $\alpha$.  We start by checking, for each node in
$\G$, whether the Hotelling $T^2$-statistic in the first $k$
graph-Fourier components of any subgraph of size $q$ containing this
node can be guaranteed to be below the level-$\alpha$ critical value
$T^2_{\alpha,k}$ (\emph{e.g.}, $(1-\alpha)$-quantile of
$\cF(k,n_1+n_2-k-1)$ distribution). If this is the case, the node is
removed from the graph. We then repeat the procedure on the edges of
the remaining graph and, iteratively, on the subgraphs up to size
$q-1$, at which point we test all the remaining subgraphs of size $q$.

Specifically, for a subgraph $g$ of $\G$ of size $q \leq p$,
Hotelling's $T^2$-statistic in the first $k \leq q$ graph-Fourier
components of $g$ is defined as
\[
\tilde{T}_k^2(g)=\frac{n_1 n_2}{n_1+n_2}(\bar{x}_1(g) -
\bar{x}_2(g))^\top U_{[k]} \left( U_{[k]}^\top \hat{\Sigma}(g) U_{[k]}
\right)^{-1}U_{[k]}^\top (\bar{x}_1(g) - \bar{x}_2(g)),
\]
where $U_{[k]}$ is the $q \times k$ restriction of the matrix of $q$ 
eigenvectors of the Laplacian of $g$ to its first $k$ columns
(\emph{i.e.}, $U_{[k]}(g)$, where we omit $g$ to ease notation) and $\bar{x}_i(g),\; i=1,2$,
and $\hat{\Sigma}(g)$ are, respectively, the empirical means and
pooled covariance matrix restricted to the nodes in $g$.  We make use
of the following upper bound on $\tilde{T}_k^2(g)$.

\begin{lemma}[Upper bound on $\tilde{T}_k^2$]
  \label{lem:neighBound}
  For any subgraph $g$ of $\G$ of size $q \leq p$, 
any subgraph $g'$ of $g$ of size $q'\leq q$, and any $k \leq q$, then
  \begin{displaymath}
    \tilde{T}_k^2(g) \leq T^2(\nu(g',q-q')) \,,
  \end{displaymath}
  where $\nu(g',r)$ is the $r$-neighborhood of $g'$, that is, the union of the nodes of $g'$ and the nodes
  whose shortest path to a node of $g'$ is less than or equal to
  $r$.
\end{lemma}
The proof involves the following result:
\begin{lemma}[Bessel inequality for Mahalanobis norm]
  \label{lem:bessel}
  Let $\Sigma \in \RR^{p,p}$ be an invertible matrix and
  $P\in\RR^{p,k}$, $k \leq p$, be a matrix with orthonormal columns. For any $x \in\RR^p$,
\[
x^\top\Sigma^{-1}x \geq
x^{\top}P\left(P^\top\Sigma P\right)^{-1}P^\top x.
\]
\end{lemma}
\begin{proof}
  First note that, by orthonormality of the columns of $P$,
  $P^\top\Sigma P$ is indeed invertible, and that
  \[
  \Sigma^{-1} - P\left(P^\top\Sigma P\right)^{-1}P^\top =
  \Sigma^{-\frac{1}{2}}\left( I - \Sigma^{\frac{1}{2}}
    P\left(P^\top\Sigma^{\frac{1}{2}}\Sigma^{\frac{1}{2}}
      P\right)^{-1}P^\top\Sigma^{\frac{1}{2}}\right)\Sigma^{-\frac{1}{2}},
  \]
  where $\Sigma^{\frac{1}{2}}
  P\left(P^\top\Sigma^{\frac{1}{2}}\Sigma^{\frac{1}{2}}
    P\right)^{-1}P^\top\Sigma^{\frac{1}{2}}$ is an orthogonal
  projection, with eigenvalues either $0$ or $1$.  Thus,
  $I-\Sigma^{\frac{1}{2}}
  P\left(P^\top\Sigma^{\frac{1}{2}}\Sigma^{\frac{1}{2}}
    P\right)^{-1}P^\top\Sigma^{\frac{1}{2}}$ is
  positive-semi-definite, as its eigenvalues are also either $0$ or
  $1$. The result follows from properties of products of
  positive-semi-definite matrices.
\end{proof}
We can now prove Lemma~\ref{lem:neighBound}.
\begin{proof}
  By Lemma~\ref{lem:bessel}, 
  \begin{align*}
  \tilde{T}_k^2(g) &\leq
  \frac{n_1n_2}{n_1+n_2}(\bar{x}_1(g) -  \bar{x}_2(g))^\top U \bigl( U^\top \hat{\Sigma}(g) U
  \bigr)^{-1}U^\top (\bar{x}_1(g) - \bar{x}_2(g)) \\
  &= \frac{n_1n_2}{n_1+n_2}(\bar{x}_1(g) - \bar{x}_2(g))^\top
  \hat{\Sigma}(g)^{-1}(\bar{x}_1(g) - \bar{x}_2(g)) = T^2(g).
\end{align*}
As $g \subset \nu(g',q-q')$, applying Lemma~\ref{lem:bessel} a second
time with the compression from $\nu(g',q-q')$ to the nodes of $g$
yields the result.
\end{proof}
Note that the bound takes into account the fact that the
$T^2$-statistic is eventually computed in the first few components of
a basis which is not known beforehand~: at each step, for each
potential subgraph $g'$ which would include the subgraph $g$ which we
consider for pruning, the $\tilde{T}_k^2(g')$ that we need to upper
bound depends on the graph Laplacian of $g'$.

\subsection{Mean-shift approximation}
\label{sec:euclidean}

\todo{*** CHECK: Add details on algo and how bound is used, relate to
  Lemma 2 bound. Strict vs. non-strict ineq for rejection, cf. Lemma
  4.}

For ``small-world'' graphs above a certain level of connectivity and
$q$ large enough, the $(q-s)$-neighborhood of $g'$, $\nu(g',q-s)$,
tends to be large, at least at the beginning of the above exact
algorithm, and the number of tests actually performed won't decrease
much compared to the total number of possible tests. One can, however,
identify much more efficiently the subgraphs whose sample mean shift
in the first $k$ components of the graph-Fourier space has Euclidean
norm $\|\hat{\tilde{\delta}}_{[k]}(g)\| = \|U_{[k]}^\top(\bar{x}_1(g)
- \bar{x}_2(g))\|$ above a certain threshold. Indeed, it is
straightforward to see that
\begin{align*}
\|U_{[k]}^\top(\bar{x}_1(g) - \bar{x}_2(g))\|^2 &\leq
\|U^\top(\bar{x}_1(g) - \bar{x}_2(g))\|^2\\ 
&= \|\bar{x}_1(g) - \bar{x}_2(g)\|^2 \\
&\leq \|\bar{x}_1(g')-\bar{x}_2(g')\|^2 \\
+ & \max_{v_1,\ldots,v_{q-s}\in\nu(g',q-s)}\|\bar{x}_1(v_1,\ldots,v_{q-s})-\bar{x}_2(v_1,\ldots,v_{q-s})\|^2.
\end{align*}
This inequality can then be used in the procedure of
Section~\ref{sec:exact}, to identify all subgraphs for which the
Euclidean norm of the sample mean shift exceeds a given threshold:
$\|\hat{\tilde{\delta}}_{[k]}(g)\|^2 > \theta$. For any $\alpha$, if
this threshold $\theta$ is low enough, all the subgraphs with
$\tilde{T}_k^2(g)>T^2_{\alpha,k}$ are included in this set. Performing
the actual $T^2$-test on these pre-selected subgraphs yields exactly
the set of subgraphs that would have been identified using the exact
procedure of Section~\ref{sec:exact}. More precisely, we have the
following result:
\begin{lemma}
\label{lem:nc}
For any threshold $\theta > 0$, $k \leq q \leq p$, and any subgraph
$g$ of size $q$ such that $\|\hat{\tilde{\delta}}_{[k]}(g)\|^2 <
\theta$,
  \begin{displaymath}  
    N\tilde{T}_k^2(g) > T^2_{\alpha,k} \Rightarrow
    \lambda_{min}(\hat{\tilde{\Sigma}}_{[k]}(g)) < \frac{Nn_1n_2\theta}{(n_1+n_2) T^2_{\alpha,k}},
  \end{displaymath}
  where $T^2_{\alpha,k}$ is the level-$\alpha$ critical value for
  $\tilde{T}_k^2$ (\emph{e.g.}, $(1-\alpha)$-quantile of
  $\cF(k,n_1+n_2-k-1)$), 
  $N = \frac{n_1+n_2-k-1}{(n_1+n_2-2)k}$ 
  and
  $\lambda_{min}(\hat{\tilde{\Sigma}}_{[k]}(g))$ denotes the smallest
  eigenvalue of $\hat{\tilde{\Sigma}}_{[k]}(g) =
  U_{[k]}\hat{\Sigma}(g)U_{[k]}^\top$.
\end{lemma}
\begin{proof}
  As $I -
  (\hat{\tilde{\Sigma}}_{[k]}(g))^{-1}\lambda_{min}(\hat{\tilde{\Sigma}}_{[k]}(g))
  \succeq 0$, it follows that, for any $x$, 
  \[
  x^\top
  (\hat{\tilde{\Sigma}}_{[k]}(g))^{-1}x \leq
  \frac{\|x\|^2}{\lambda_{min}(\hat{\tilde{\Sigma}}_{[k]}(g))}.
  \]
\end{proof}

Lemma~\ref{lem:nc} states that for any subgraph which would be detected by Hotelling's
$T^2$-statistic $\tilde{T}_k^2(g)$ but not by the Euclidean criterion
$\|\hat{\tilde{\delta}}_{[k]}(g)\|^2$, the sample covariance matrix in the
restricted graph-Fourier space has an eigenvalue below a certain
threshold. This implies that such false
negative subgraphs (from the Euclidean approximation to the exact
algorithm) always have a small mean shift in the graph-Fourier space,
but in a direction of small variance.  In context of gene expression,
this is related to the well-known issue of the detection of DE genes
by virtue of their small variances. Even though the differences in
expression appear to be highly significant for these genes, they
correspond to small effects that are not interesting from a practical
point of view (\emph{i.e.}, biologically insignificant). \todo{PN: elaborate on how such cases could be caused by artifacts in the assay or pre-processing} Methods for
addressing this problem are proposed in
~\cite{Lonnstedt2001Replicated}.  Note that
$\lambda_{min}(\hat{\Sigma}(g))) \leq
\lambda_{min}(\hat{\tilde{\Sigma}}_{[k]}(g)))$; thus, the remark on
variances holds for both the graph-Fourier and original
spaces. However, if $q$ is large, we expect
$\lambda_{min}(\hat{\Sigma}(g))$ to be very small, while filtering
somehow controls the conditioning of the covariance matrix.

\subsection{Multiple testing}

\todo{LJ: Provide examples of the number of investigated
  subgraphs and number of permutations.}

\todo{SD: Discuss choice of Type I error rate and complete null
  hypothesis. More precise notation. Deserve further study. Ongoing
  questions.}

Testing for homogeneity over the potentially large number of subgraphs
investigated as part of the above algorithms immediately raises the
issue of multiple testing. However, the present multiplicity problem
is unusual, in the sense that one does not know in advance the total
number of tests and which tests will be performed
specifically. Standard multiple testing procedures, such as those in
\cite{Dudoit2008Multiple}, are therefore not immediately applicable.

In an attempt to address the multiplicity issue, we apply a
permutation procedure to control the number of false positive
subgraphs under the complete null hypothesis of identical
distributions in the two populations.  Specifically, one permutes the
class/population labels (1 or 2) of the $n_1+n_2$ observations and
applies the non-homogeneous subgraph discovery algorithm to the
permuted data to yield a certain number of false positive
subgraphs. Repeating this procedure a sufficiently large number of
times produces an estimate of the distribution of the number of Type I
errors under the complete null hypothesis of identical distributions.

\section{Results}
\label{sec:experiments}

We evaluate the empirical behavior of the procedures proposed in
Sections~\ref{sec:test} and \ref{sec:discovery}, first on synthetic
data, then on breast cancer microarray data analyzed in context of
KEGG pathways.

\subsection{Synthetic data}

\todo{*** CHECK: Add details: Random graph, exact value of $\delta$,
  $\Sigma$, $k_0$, $k$, power, Type I error rate, etc.}

The performance of the graph-structured test is assessed in cases
where the distribution shift $\Delta^2$ satisfies the smoothness
assumptions described in Section~\ref{sec:test}.  We first generate a
connected random graph with $p=20$ nodes.  Next, we generate $10,000$
datasets, each comprising $n_1=n_2=20$ Gaussian random vectors in
$\RR^{p}$, with null mean shift $\delta$ for $5,000$ datasets and mean
shift $\delta=1$ for the remaining $5,000$. For the latter datasets, the
non-zero shift is built in the first $k_0=3$ Fourier coefficients (the
shift being zero for the remaining $p-k_0$ coefficients) and an
inverse Fourier transformation is applied to random vectors generated
in the graph-Fourier space. We consider two covariance settings: in
the first one, the covariance matrix in the graph-Fourier space is
diagonal with diagonal elements at $\frac{1}{\sqrt{p}}$. In the second
one, correlation is introduced between the shifted coefficients
only. Specifically, for $i,j\leq k_o$,
$\Sigma_{ij}=\frac{0.5}{\sqrt{p}}$ if $i\neq j$,
$\Sigma_{ii}=\frac{0.9}{\sqrt{p}}$ otherwise.

Figure~\ref{fig:rocs} displays receiver operator characteristic (ROC)
curves for mean shift detection by the standard Hotelling $T^2$-test,
$T^2$ in the first $k_0$ Fourier coefficients, $T^2$ in the first
$k_0$ principal components (PC), the adaptive Neyman test of
~\cite{Fan1998Test}, and a modified version of this test 
where the correct value of $k_0$ is specified. Note that we do not
consider sparse learning
approaches~\citep{Jacob2009Group,Jenatton2009Structured}, but it would
be straightforward to design a realistic setting where such approaches
are outperformed by testing, \emph{e.g.}, by adding correlation
between some of the functions under $\HO$.

\begin{figure}[h]
  \begin{center}
    \includegraphics[width=0.45\linewidth]{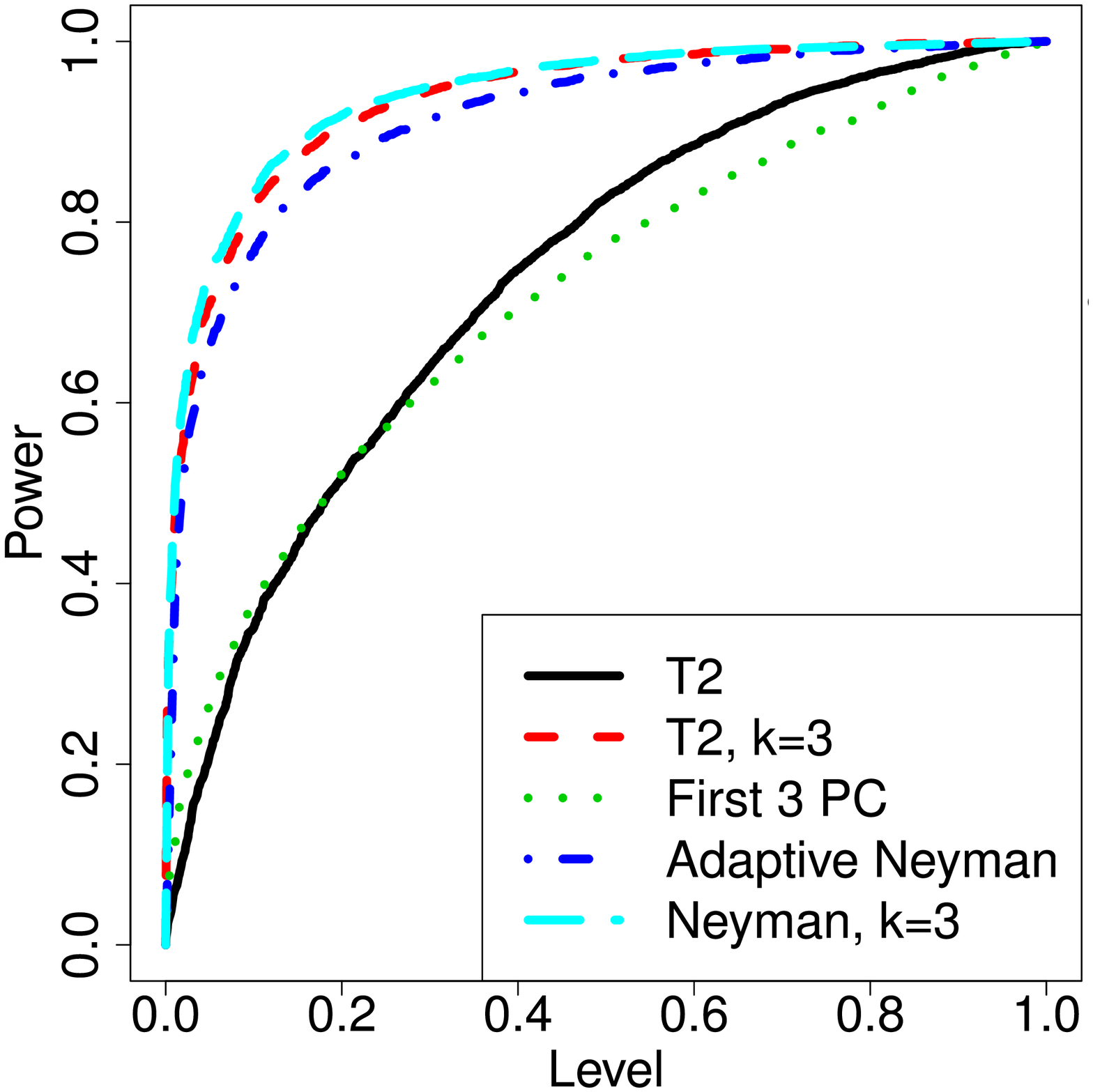}
    \includegraphics[width=0.45\linewidth]{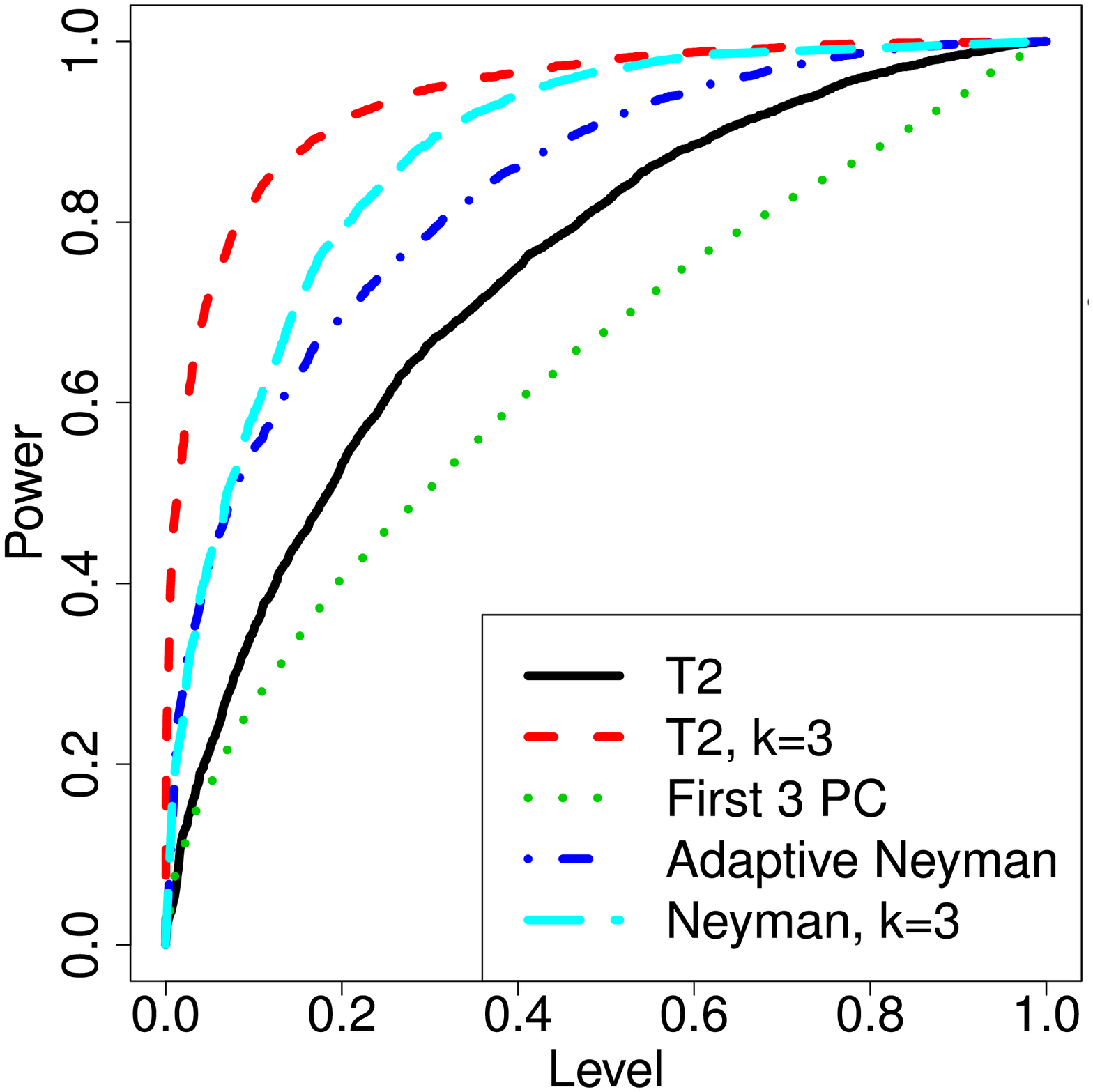}
  \end{center}
  \caption{ROC curves for the detection of a smooth shift using various
    test statistics. Left: Diagonal covariance structure. Right:
    Block-diagonal covariance structure. }
  \label{fig:rocs}
\end{figure}
\todo{*** CHECK: Figure \ref{fig:rocs}: Power vs. Type I error rate, legend, lty.}
The first important comparison is between the classical Hotelling
$T^2$-test versus the $T^2$-test in the graph-Fourier space. As
expected from Lemma~\ref{lem:das}, testing in the restricted space
where the shift lies performs much better than testing in the full
space which includes irrelevant coefficients. The difference can be
made arbitrarily large by increasing the dimension $p$ and keeping the
shift unchanged. The graph-structured test retains a large advantage
even for moderately smooth shifts, \emph{e.g.}, when $k_0=3$ and
$p=5$. Of course, this corresponds to the optimistic case where the
number of shifted coefficients $k_0$ is known. Figure~\ref{fig:miss}
shows the power of the test in the graph-Fourier space for various
choices of $k$. Even when missing some coefficients ($k<k_0$) or adding a few
non-relevant ones ($k>k_0$), the power of the graph-structured test is higher than
that of the $T^2$-test in the full space. The principal component
approach is shown because it was proposed for the application which
motivated our work~\citep{Ma2009Identification} and as it
illustrates that the performance improvement originates not only from
dimensionality reduction, but also from the fact that this reduction
is in a direction that does not decrease the shift. We emphasize that
power entirely depends on the nature of the shift and that a PC-based
test would outperform our Fourier-based test when the shift lies in
the first principal components rather than Fourier coefficients. The
statistics of~\cite{Bai1996Effect} and \cite{Chen2010A} are also
largely outperformed by our graph-structured statistic (ROC curves not
shown in Figure~\ref{fig:rocs} for the sake of readability), which
illustrates that working in the graph-Fourier space solves the problem
of high-dimensionality for which these statistics were designed. Here
again, for a non-smooth shift, the comparison would be less
favorable. Finally, we consider the adaptive Neyman
test of \cite{Fan1998Test}, which takes advantage of smoothness assumptions for
time-series. This test differs from our graph-structured test, as
Fourier coefficients for stationary time-series are known to be
asymptotically independent and Gaussian. For graphs, the asymptotics
would be in the number of nodes, which is typically small, and
necessary conditions such as stationarity are more difficult to define
and unlikely to hold for data like gene expression measurements. In the
uncorrelated setting, the modified version of the 
\cite{Fan1998Test} statistic based the true number of non-zero coefficients performs
approximately as well as the graph-structured $T^2$. However, for
correlated data, it loses power and both versions
can have arbitrarily degraded performance. This, together with the
need to use the bootstrap to calibrate this test,
illustrates that direct transposition of the \cite{Fan1998Test} test
to the graph context is not optimal.

\OMIT{%
  The problem is more complex for graphs, because the asymptotic would
  be in the number of nodes, which is much smaller than the number of
  points in a time-series, and the necessary conditions such as
  stationarity are more difficult to define and unlikely for data like
  gene expressions. In the uncorrelated setting, the Neyman statistic
  of ~\cite{Fan1998Test} using the correct number of coefficients
  performs approximately as well as the $T^2$. The adaptive version
  loses some power, and both version can have arbitrarily degraded
  performances when adding correlation. This together with the need to
  use bootstrap to calibrate the adaptive Neyman statistic illustrate
  that direct transposition of~\cite{Fan1998Test} to the context of
  graphs is not optimal.  }

\begin{figure}[h]
  \begin{center}
    \includegraphics[width=0.5\linewidth]{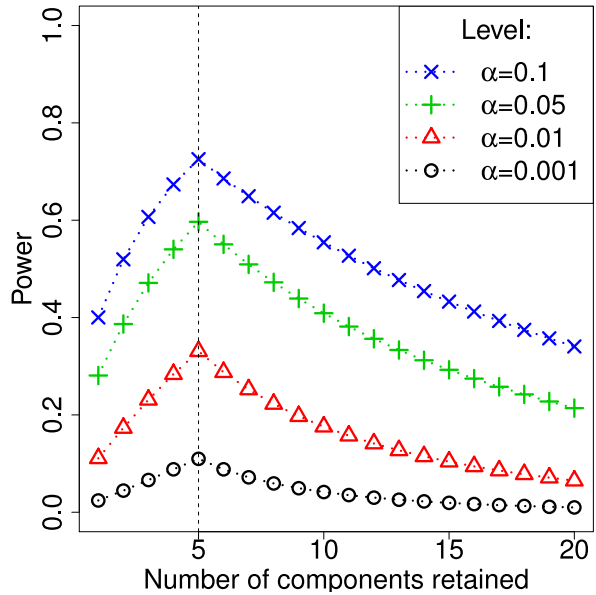}
  \end{center}
  \caption{Power of the $T^2$-test in the graph-Fourier space with an
    actual mean shift evenly distributed among the first $k_0=5$
    coefficients. }
  \label{fig:miss}
\end{figure}

\todo{*** CHECK: Add details: Random graph, exact value of $\delta$,
  $\Sigma$, $k_0$, $k$, $q$, $n_i$, $\alpha$, number of permutations.
  Check calc of $\lambda_{min}$, should report $\lambda_{min} <$
  according to Lemma 4.  Machine specs for run-time. Table for
  run-time with ratios, number of visited subgraphs?}  To evaluate the
performance of the subgraph discovery algorithms proposed in
Section~\ref{sec:discovery}, we generated a graph of $100$ nodes
formed by tightly-connected hubs of sizes sampled from a Poisson distribution with parameter 10 and only weak connections
between these hubs (Figure~\ref{fig:randg}). Such a graph structure
mimics the typical topology of gene regulation networks. We randomly
selected one subgraph of $5$ nodes to be non-homogeneous, with smooth
shift in the first $k_0=3$ Fourier coefficients. The mean shift was
set to zero on the rest of the graph. We set the norm of the mean
shift to $1$ and the covariance matrix to identity, so that detecting
the shifted subgraph is impossible by just looking at the mean shift
on the graph.

\begin{figure}[h]
  \begin{center}
      \includegraphics[width=.6\linewidth]{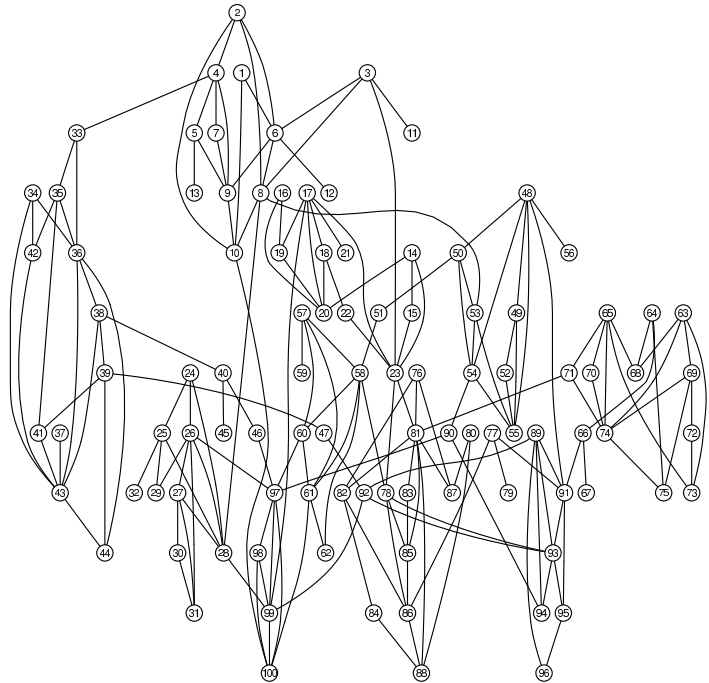}
  \end{center}
  \caption{Random graph used in the evaluation of the pruning procedure.}
  \label{fig:randg}
\end{figure}

We evaluated run-time for full enumeration, the exact branch-and-bound
algorithm based on Lemma~\ref{lem:neighBound} (Section
\ref{sec:exact}), and the approximate algorithm based on the Euclidean
norm (Section~\ref{sec:euclidean}). We also examined run-time on data
with permuted class labels, as the subgraph discovery procedure is to
be run on such data to evaluate the number of false positives and
adjust for multiple testing. Averaging over $20$ runs, the full
enumeration procedure took $732\pm 9$ seconds per run and the exact
branch-and-bound $627\pm 59$ seconds on the non-permuted data and
$578\pm 100$ seconds on permuted data.  Over $100$ runs, the
approximation at $\theta=0.5$ ($\lambda_{min}=0.52$) took $204\pm 86$
seconds ($129\pm 91$ on permuted data) and the approximation at
$\theta=1$ ($\lambda_{min}=1.04$) took $183\pm 106$ seconds ($40\pm
60$ on permuted data). The latter approximation missed the
non-homogeneous subgraph in $5\%$ of the runs.

\OMIT{%
Table~\ref{tab:bbsyn} shows the run-times for the approximation
approach at various thresholding levels, and corresponding values of
the lower bound on the smallest covariance eigenvalue for false
negatives.
\begin{table}[t]
\caption{Average run-time of testing all the subgraphs on synthetic
  data at various levels.}
\label{tab:bbsyn}
\begin{center}
\begin{tabular}{lccccc}
$\lambda_{min}$ & $0.69$ & $1$ & $1.1$ & $1.25$ & $1.5$
\\ \hline \\
$\theta$ & $0.5$ & $0.72$ & $0.79$ & $0.9$ & $1.08$\\
Run-time         & $204\pm 86$ & $121\pm 72$ & $89\pm 61$ & $54\pm 41$
& $39\pm34$\\
Run-time Perm.   & $129\pm 91$ & $43\pm 57$ & $15\pm 29$ & $6\pm 13$ & $4\pm 5$\\
\end{tabular}
\end{center}
\end{table}
}

While neither the exact nor the approximate bounds are efficient
enough to allow systematic testing on huge graphs for which the exact
approach would be impossible, they allow a significant gain in speed,
especially for permuted data, and will thus prove to be very useful
for multiple testing adjustment.

\subsection{Breast cancer expression data}

\todo{*** CHECK: KEGG version.}

We also validated our methods using the microarray dataset
of~\cite{Loi2008Predicting}, which comprises
expression measures for $15,737$ genes in $255$ patients treated with tamoxifen.  Using distant metastasis free survival as a primary endpoint, $68$ patients are labeled as resistant to tamoxifen and $187$ are labeled as sensitive to tamoxifen. 
Our goal was to detect structured groups of genes which are
differentially expressed between resistant and sensitive patients. 

\todo{report results of univariate DE analysis for comparison ?}

We first tested individually $323$ connected components from $89$ KEGG
pathways corresponding to known gene regulation networks, using the
classical Hotelling $T^2$-test and the $T^2$-test in the graph-Fourier
space retaining only the first $20\%$ Fourier coefficients ($k=0.2 p$). For each of the
$323$ graphs, (unadjusted) $p$-values were computed under the nominal
$F$-distributions $\cF(p, n_1+n_2-p-1)$ and $\cF(k, n_1+n_2-k-1)$,
respectively. Figure~\ref{fig:path} shows the pathway for which the
ratio of graph-Fourier to full space $p$-values is the lowest
(\emph{i.e.}, most significant for graph-structured test relative to
classical test) and the pathway for which it is the highest. As
expected, the former corresponds to a shift which appears to be
coherent with the network (even on edges corresponding to inhibition),
while the latter is a small network with non-smooth shift. More
generally, the classical approach tends to select very small
networks. The coherent pathway selected by our graph-structured test
corresponds to \emph{Leukocyte transendothelial migration}. To the
best of our knowledge, this pathway is not specifically known to be
involved in tamoxifen resistance. However, its role in resistance is
plausible, as leukocyte infiltration was recently found to be involved
in breast tumor invasion~\citep{Man2010Aberrant}; more generally, the
immune system and inflammatory response are closely-related to the
evolution of cancer.

\begin{figure}[h]
  \begin{center}
      \includegraphics[width=.8\linewidth]{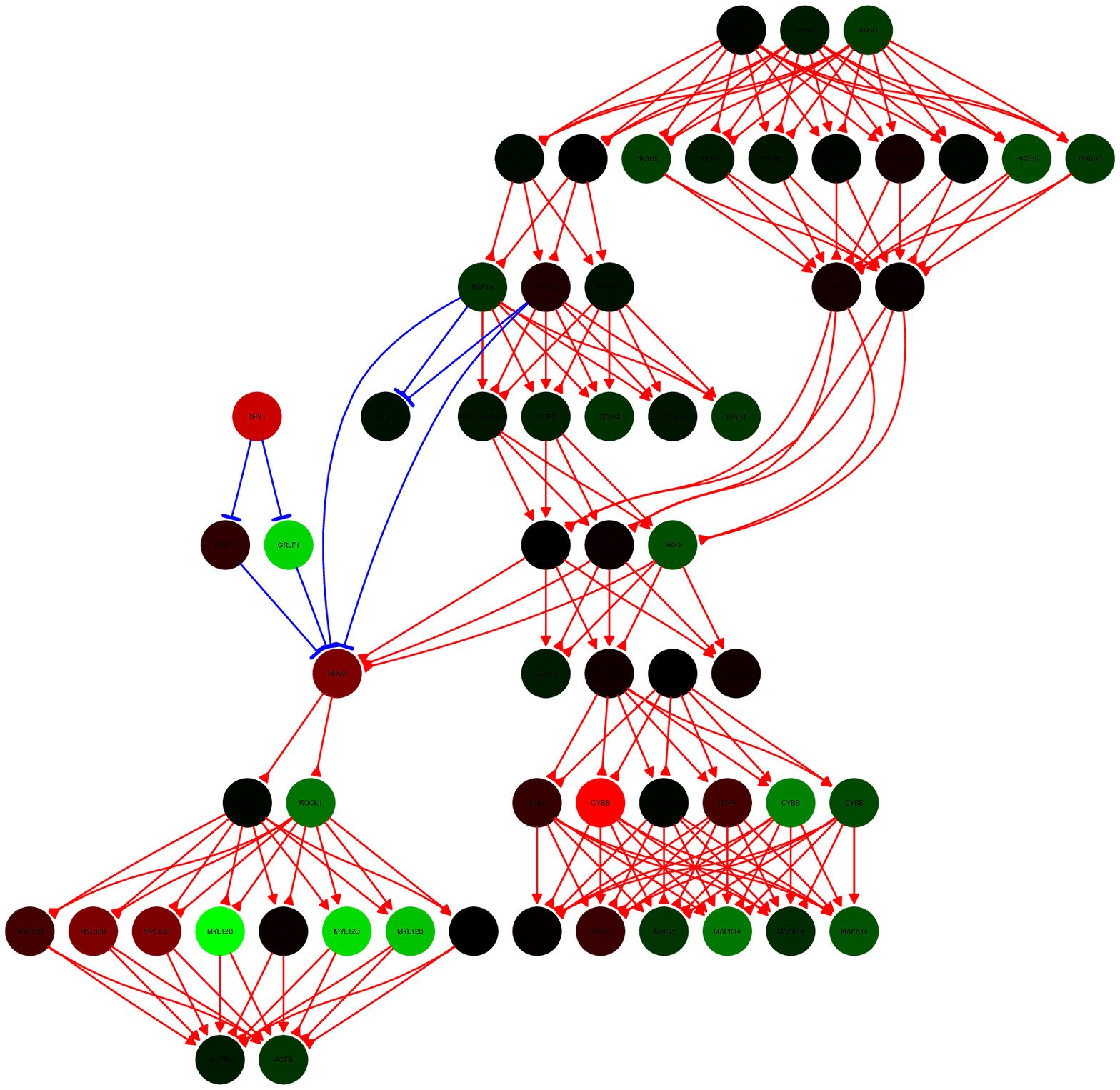}
      \includegraphics[width=.3\linewidth]{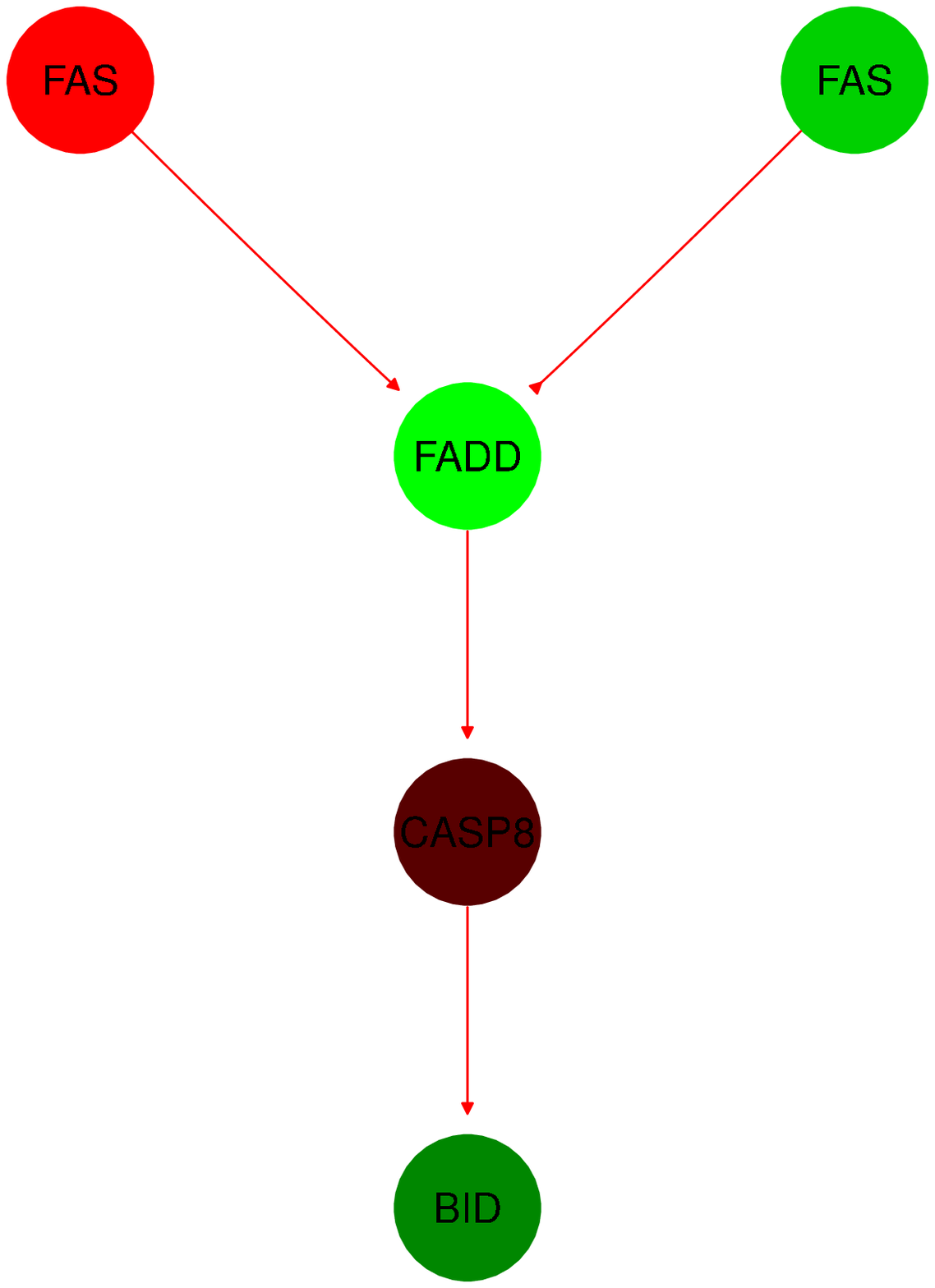}
  \end{center}
  \caption{Difference in sample mean expression measures between
    tamoxifen-resistant and sensitive patients, for genes in two KEGG
    regulation networks. Top: Regulation network (Leukocyte
    transendothelial migration) with the lowest ratio of graph-Fourier
    to full space $p$-values. Bottom: Regulation network (Alzheimer's
    disease) with the highest ratio of graph-Fourier to full space
    $p$-values. Nodes are colored according to the value of the
    difference in means, with green corresponding to high positive
    values, red to high negative values, and black to $0$. Red arrows
    denote activations, blue arrows inhibition.}
  \label{fig:path}
\end{figure}
\todo{*** CHECK: Add details: Exact value of $k$, $q$, $\alpha$,
  number of permutations.  Squared norm? Check calc of
  $\lambda_{min}$.  $\chi^2_5 \simeq 0.5$? $\alpha=2.10^{-4}$?}  

We then ran our branch-and-bound non-homogeneous subgraph discovery
procedure on the cell cycle pathway, which, after restriction to edges
of known sign (inhibition or activation), has $86$ nodes and $442$
edges. Specifically, we sought to detect differentially expressed
subgraphs of size $q=5$, after pre-selecting those for which the
squared Euclidean norm of the empirical shift exceeded $\theta=0.1$;
for a test in the first $k=3$ Fourier components at level
$\alpha=10^{-4}$, this corresponded to $\lambda_{min}<0.23$ and to an
expected removal of $95\%$ of the subgraphs under the approximation
that the squared Euclidean norm of the subgraphs follows a
$\chi^2_5$-distribution. 

For $\alpha=10^{-4}$, none of the $50$ runs on permuted data gave any
positive subgraph and $31$ overlapping subgraphs
(Figure~\ref{fig:firstGraphSign}) were detected on the original data,
corresponding to a connected subnetwork of $22$ genes. Some of these
genes have large individual differential expression, namely TP53 whose
mutation has been long-known to be involved in tamoxifen resistance
\citep{Andersson2005Worse,Fernandez-Cuesta2010p53}. E2F1, whose
expression level was recently shown to be involved in tamoxifen
resistance \citep{Louie2010Estrogen}, is also part of the identified
network, as well as CCND1
\citep{Barnes1997Cyclin,Musgrove2009Biological}. Some other genes in
the network have quite low $t$-statistics and would not have been
detected individually. This is the case of CCNE1 and CDK2, which were
also described in \citep{Louie2010Estrogen} as part of the same
mechanism as E2F1. Similarly, CDKN1A was recently found to be involved
in anti-\oe strogene treatment resistance
\citep{Musgrove2009Biological} and in ovarian cancer which is also a
hormone-dependent cancer \citep{Cunningham2009Cell}. The networks also
contains RB1, a tumor suppressor whose expression or loss is known to
be correlated to tamoxifen resistance
\citep{Musgrove2009Biological}. RB1 is inhibited by CDK4, whose
inhibition has been described in \cite{Sutherland2009CDK} as acting
synergistically with tamoxifen or trastuzumab. More generally, a large
part of the network displayed on Figure~2A of
\cite{Musgrove2009Biological} is included in our network, along with
other known actors of tamoxifen resistance. Our system-based approach
to pathway discovery therefore directly identifies a set of
interacting important genes and may therefore prove to be more
efficient than iterative individual identification of single actors.

\OMIT{
  Note that for $\theta=0.04$ ($\lambda_{min}=0.09$, $\theta$
  corresponding to $\approxeq 50\%$ of $\chi^2_5$), $25,000$ graphs
  are already pre-selected on permuted data, none of which reaches the
  critical level.  For $\alpha=10^{-4}$, none of the $50$ runs on
  permuted data gave any positive subgraph, and two overlapping
  subgraphs (Figure~\ref{fig:firstsg}) were detected on the original
  data, with one gene long known to be involved in tamoxifen resistant
  (TP53, in \cite{Andersson2005Worse}), one very recently discovered
  (E2F1, in \cite{Louie2010Estrogen}) and in the middle of the
  subgraph, some less known (CDKNA1-2) with quite low individual
  $t$-scores, which were recently found to be involved in ovarian
  cancer which is also an hormone-dependent cancer
  \citep{Cunningham2009Cell}. For $\alpha=2.10^{-4}$, only two of $50$
  permuted runs selected $2$ subgraphs each, and $15$ subgraphs, whose
  union actually forms a unique connected component were detected as
  significant (Figure~\ref{fig:secsg}).
}%

\begin{figure}[h]
  \begin{center}
      \includegraphics[width=.8\linewidth]{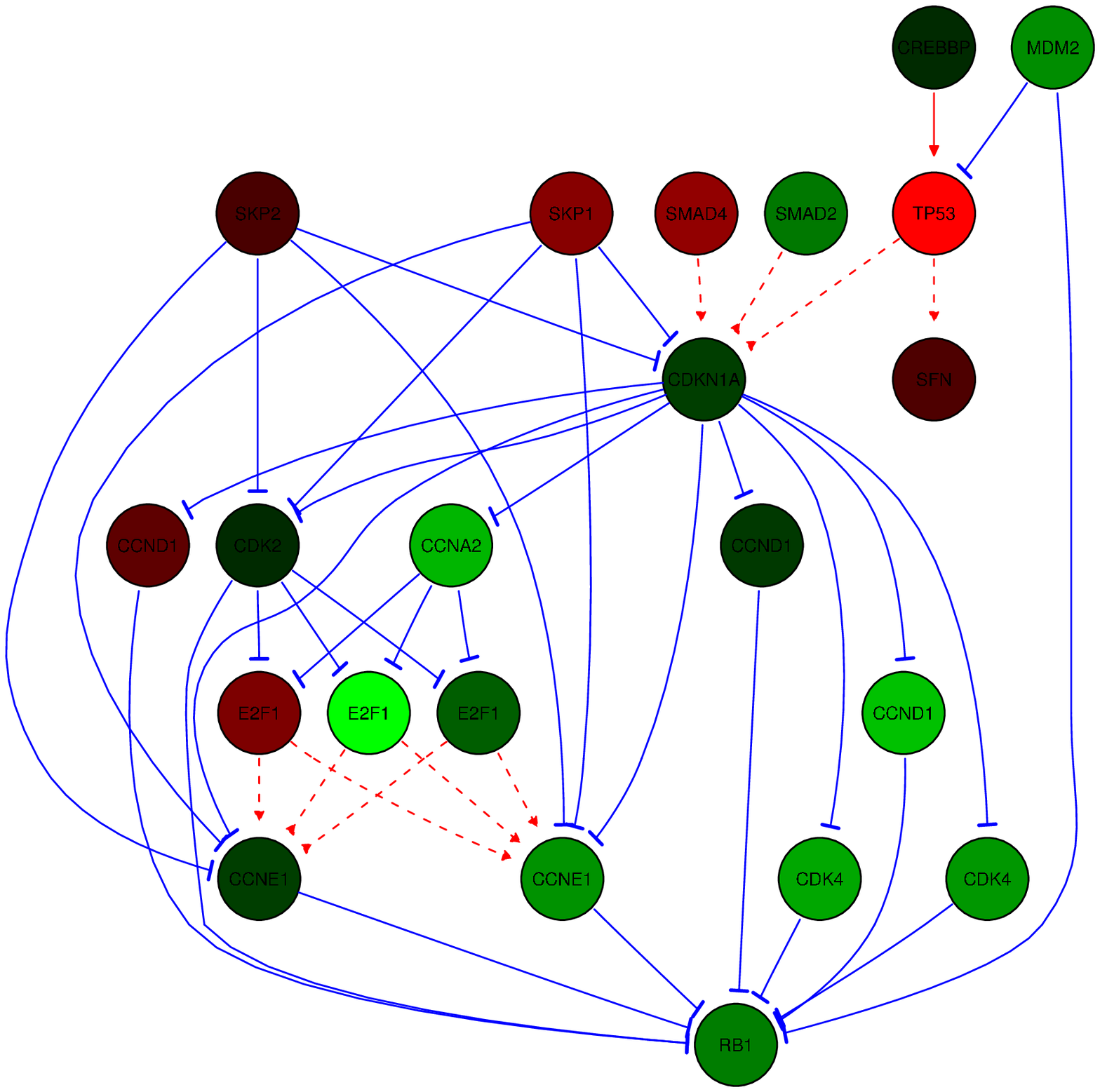}
  \end{center}
  \caption{Difference in sample mean expression measures between
    tamoxifen-resistant and sensitive patients, for genes in the two
    overlapping subgraphs detected at $\alpha=10^{-4}$. Nodes are
    colored according to the value of the difference in means, with
    green corresponding to high positive values, red to high negative
    values, and black to $0$. Red arrows denote activations, blue
    arrows inhibition.}
  \label{fig:firstGraphSign}
\end{figure}


\section{Discussion}
\label{sec:discussion}

We developed a graph-structured two-sample test of means, for problems
in which the distribution shift is assumed to be smooth on a given
graph. We proved quantitative results on power gains for such
smooth-shift alternatives and devised branch-and-bound algorithms to
systematically apply our test to all the subgraphs of a large
graph. The first algorithm is exact and reduces the number of
explicitly tested subgraphs. The second is approximate, with no false
positives and a quantitative result on the type of false negatives
(with respect to the exact algorithm). The non-homogeneous subgraph
discovery method involves performing a larger number of tests, with
highly-dependent test statistics.  However, as the actual number of
tested hypotheses is unknown, standard multiple testing procedures are
not directly applicable. Instead, we use a permutation procedure to
estimate the distribution of the number of false positive
subgraphs. Such resampling procedures (bootstrap or permutation) are
feasible due to the manageable run-time of the pruning algorithms of
Section \ref{sec:discovery}. Results on synthetic data illustrate the
good power properties of our graph-structured test under smooth-shift
alternatives, as well as the good performance of our
branch-and-bound-like algorithms for subgraph discovery. Very
promising results are also obtained on the drug resistance microarray
dataset of \cite{Loi2008Predicting}.

Future work should investigate the use of other bases, such as
graph-wavelets~\citep{Hammond2009Wavelets}, which would allow the
detection of shifts with spatially-located non-smoothness, for
example, to take into account errors in existing networks. More
systematic procedures for cutoff selection should also be considered,
\emph{e.g.}, two-step method proposed in~\cite{Das1974Power} or
adaptive approaches as in~\cite{Fan1998Test}. The pruning algorithm
would naturally benefit from sharper bounds. Such bounds could be
obtained by controlling the condition number of all covariance
matrices, using, for example, regularized statistics which still have
known non-asymptotic distributions, such
as those of~\cite{Tai2008On}. Concerning multiple testing, procedures should be
devised to exploit the dependence structure between the tested
subgraphs and to deal with the unknown number of tests. The proposed
approach could also be enriched to take into account different types
of data, \emph{e.g.}, copy number for the detection of DE gene
pathways. More subtle notions of smoothness, \emph{e.g.}, ``and'' and
``or'' logical relations~\citep{Vaske2010Inference}, could also be
included. An interesting alternative application would be to explore
the list of pathways which are known to be differentially expressed
(or detected by the classical $T^2$-test), but which are not detected
by the graph-Fourier approach, to infer possible mis-annotation in the
network. Other applications of two-sample tests with smooth-shift on a
graph include fMRI and eQTL association studies.

Finally, it would be of interest to compare our testing approach with
structured sparse learning, for the purpose of identifying expression
signatures that are predictive of drug resistance.  Methods should be
compared in terms of prediction accuracy and stability of the selected
genes across different datasets, a central and difficult problem in
the design of such
signatures~\citep{Ein-Dor2005Outcome,He2010Stable,Haury2010Increasing}. The
comparison should also take into account the merits of the
sparsity-inducing norm over the hypothesis testing-based selection, as
well as the influence of the smoothness assumption. The latter could
indeed also be integrated in a sparsity-inducing penalty by applying,
\emph{e.g.}, \cite{Jacob2009Group} \todo{*** CHECK: Unclear}
to the reduced graph-Fourier representation of the pathways, yielding
a special case of multiple kernel learning~\citep{Bach2004Multiple}.

\subsubsection*{Acknowledgments}

The authors thank Za\"{i}d Harchaoui, Nourredine El Karoui, and Terry
Speed for very helpful discussions and suggestions, and the UC
Berkeley Center for Computational Biology Genentech Innovation
Fellowship and The Cancer Genome Atlas Project for funding.

\bibliographystyle{plainnat}

\begin{thebibliography}{40}
\providecommand{\natexlab}[1]{#1}
\providecommand{\url}[1]{\texttt{#1}}
\expandafter\ifx\csname urlstyle\endcsname\relax
  \providecommand{\doi}[1]{doi: #1}\else
  \providecommand{\doi}{doi: \begingroup \urlstyle{rm}\Url}\fi

\bibitem[Andersson et~al.(2005)Andersson, Larsson, Klaar, Holmberg, Nilsson,
  Inganäs, Carlsson, Ohd, Rudenstam, Gustavsson, and
  Bergh]{Andersson2005Worse}
J.~Andersson, L.~Larsson, S.~Klaar, L.~Holmberg, J.~Nilsson, M.~Inganäs,
  G.~Carlsson, J.~Ohd, C-M. Rudenstam, B.~Gustavsson, and J.~Bergh.
\newblock Worse survival for tp53 (p53)-mutated breast cancer patients
  receiving adjuvant cmf.
\newblock \emph{Ann Oncol}, 16\penalty0 (5):\penalty0 743--748, May 2005.
\newblock \doi{10.1093/annonc/mdi150}.
\newblock URL \url{http://dx.doi.org/10.1093/annonc/mdi150}.

\bibitem[Bach et~al.(2004)Bach, Lanckriet, and Jordan]{Bach2004Multiple}
F.~R. Bach, G.~R.~G. Lanckriet, and M.~I. Jordan.
\newblock Multiple kernel learning, conic duality, and the {SMO} algorithm.
\newblock In \emph{ICML '04: Proceedings of the twenty-first international
  conference on Machine learning}, page~6, New York, NY, USA, 2004. ACM.
\newblock \doi{http://doi.acm.org/10.1145/1015330.1015424}.

\bibitem[Bai and Saranadasa(1996)]{Bai1996Effect}
Zhidong Bai and Hewa Saranadasa.
\newblock Effect of high dimension~: by an example of a two sample problem.
\newblock \emph{Statistica Sinica}, 6:\penalty0 311,329, 1996.

\bibitem[Barnes(1997)]{Barnes1997Cyclin}
D.~M. Barnes.
\newblock Cyclin d1 in mammary carcinoma.
\newblock \emph{J Pathol}, 181\penalty0 (3):\penalty0 267--269, Mar 1997.
\newblock \doi{3.0.CO;2-X}.
\newblock URL \url{http://dx.doi.org/3.0.CO;2-X}.

\bibitem[Beissbarth and Speed(2004)]{Beissbarth2004GOstat}
Tim Beissbarth and Terence~P Speed.
\newblock Gostat: find statistically overrepresented gene ontologies within a
  group of genes.
\newblock \emph{Bioinformatics}, 20\penalty0 (9):\penalty0 1464--1465, Jun
  2004.
\newblock \doi{10.1093/bioinformatics/bth088}.
\newblock URL \url{http://dx.doi.org/10.1093/bioinformatics/bth088}.

\bibitem[Chen and Qin(2010)]{Chen2010A}
Song~Xi Chen and Ying-Li Qin.
\newblock A two-sample test for high-dimensional data with applications to
  gene-set testing.
\newblock \emph{Ann. Stat.}, 38\penalty0 (arXiv:1002.4547.
  IMS-AOS-AOS716):\penalty0 808--835, Feb 2010.

\bibitem[Chung(1997)]{Chung1997Spectral}
F.~R.~K. Chung.
\newblock \emph{Spectral graph theory}, volume~92 of \emph{CBMS Regional
  Conference Series}.
\newblock American Mathematical Society, Providence, 1997.

\bibitem[Cunningham et~al.(2009)Cunningham, Vierkant, Sellers, Phelan, Rider,
  Liebow, Schildkraut, Berchuck, Couch, Wang, Fridley, Consortium,
  Gentry-Maharaj, Menon, Hogdall, Kjaer, Whittemore, DiCioccio, Song, Gayther,
  Ramus, Pharaoh, and Goode]{Cunningham2009Cell}
J.~M. Cunningham, R.~A. Vierkant, T.~A. Sellers, C.~Phelan, D.~N. Rider,
  M.~Liebow, J.~Schildkraut, A.~Berchuck, F.~J. Couch, X.~Wang, B.~L. Fridley,
  Ovarian Cancer~Association Consortium, A.~Gentry-Maharaj, U.~Menon,
  E.~Hogdall, S.~Kjaer, A.~Whittemore, R.~DiCioccio, H.~Song, S.~A. Gayther,
  S.~J. Ramus, P.~D~P Pharaoh, and E.~L. Goode.
\newblock Cell cycle genes and ovarian cancer susceptibility: a tagsnp
  analysis.
\newblock \emph{Br J Cancer}, 101\penalty0 (8):\penalty0 1461--1468, Oct 2009.
\newblock \doi{10.1038/sj.bjc.6605284}.
\newblock URL \url{http://dx.doi.org/10.1038/sj.bjc.6605284}.

\bibitem[Das~Gupta and Perlman(1974)]{Das1974Power}
Somesh Das~Gupta and Michael~D. Perlman.
\newblock Power of the noncentral {F} test~: effect of additional variates on
  hotelling's $t^2$-test.
\newblock \emph{Journal of the American Statistical Association}, 69\penalty0
  (345):\penalty0 174--180, Mar 1974.

\bibitem[Dudoit and van~der Laan(2008)]{Dudoit2008Multiple}
S.~Dudoit and M.~J. van~der Laan.
\newblock \emph{Multiple Testing Procedures with Applications to Genomics}.
\newblock Springer Series in Statistics. Springer, New York, 2008.

\bibitem[Ein-Dor et~al.(2005)Ein-Dor, Kela, Getz, Givol, and
  Domany]{Ein-Dor2005Outcome}
Liat Ein-Dor, Itai Kela, Gad Getz, David Givol, and Eytan Domany.
\newblock Outcome signature genes in breast cancer: is there a unique set?
\newblock \emph{Bioinformatics}, 21\penalty0 (2):\penalty0 171--178, 2005.

\bibitem[Fan and Lin(1998)]{Fan1998Test}
Jianqing Fan and Sheng-kuei Lin.
\newblock Test of significance when data are curves.
\newblock \emph{J. Am. Statist. Assoc}, 93:\penalty0 1007--1021, 1998.

\bibitem[Fernandez-Cuesta et~al.(2010)Fernandez-Cuesta, Anaganti, Hainaut, and
  Olivier]{Fernandez-Cuesta2010p53}
Lynnette Fernandez-Cuesta, Suresh Anaganti, Pierre Hainaut, and Magali Olivier.
\newblock p53 status influences response to tamoxifen but not to fulvestrant in
  breast cancer cell lines.
\newblock \emph{Int J Cancer}, Jun 2010.
\newblock \doi{10.1002/ijc.25512}.
\newblock URL \url{http://dx.doi.org/10.1002/ijc.25512}.

\bibitem[Goeman and B{\"u}hlmann(2007)]{Goeman2007Analyzing}
J~J Goeman and P~B{\"u}hlmann.
\newblock Analyzing gene expression data in terms of gene sets: methodological
  issues.
\newblock \emph{Bioinformatics}, 23\penalty0 (8):\penalty0 980--987, April
  2007.
\newblock \doi{10.1093/bioinformatics/btm051}.
\newblock URL \url{http://www.ncbi.nlm.nih.gov/pubmed/17303618}.

\bibitem[Goldberg(2007)]{Goldberg2007Dissimilarity}
Andrew~B. Goldberg.
\newblock Dissimilarity in graph-based semisupervised classification.
\newblock In \emph{Eleventh International Conference on Artificial Intelligence
  and Statistics (AISTATS)}, 2007.

\bibitem[Gretton et~al.(2007)Gretton, Borgwardt, Rasch, Schölkopf, and
  Smola]{Gretton2007A}
Arthur Gretton, Karsten~M. Borgwardt, Malte Rasch, Bernhard Schölkopf, and
  Alexander~J. Smola.
\newblock A kernel method for the two-sample-problem.
\newblock In B.~Sch\"{o}lkopf, J.~Platt, and T.~Hoffman, editors,
  \emph{Advances in Neural Information Processing Systems 19}, pages 513--520.
  MIT Press, Cambridge, MA, 2007.

\bibitem[Hammond et~al.(2009)Hammond, Vandergheynst, and
  Gribonval]{Hammond2009Wavelets}
David~K. Hammond, Pierre Vandergheynst, and R\'{e}mi Gribonval.
\newblock Wavelets on graphs via spectral graph theory.
\newblock \emph{CoRR}, abs/0912.3848, 2009.
\newblock URL
  \url{http://dblp.uni-trier.de/db/journals/corr/corr0912.html#abs-0912-3848}.
\newblock informal publication.

\bibitem[Harchaoui et~al.(2007)Harchaoui, Bach, and
  Moulines]{Harchaoui2008Testing}
Za\"{i}d Harchaoui, Francis Bach, and Eric Moulines.
\newblock Testing for homogeneity with kernel {F}isher discriminant analysis.
\newblock In John~C. Platt, Daphne Koller, Yoram Singer, and Sam~T. Roweis,
  editors, \emph{NIPS}. MIT Press, 2007.

\bibitem[{Haury} et~al.(2010){Haury}, {Jacob}, and {Vert}]{Haury2010Increasing}
{A.-C.} {Haury}, L.~{Jacob}, and {J.-P.} {Vert}.
\newblock {Increasing stability and interpretability of gene expression
  signatures}.
\newblock \emph{ArXiv e-prints}, January 2010.

\bibitem[He and Yu(2010)]{He2010Stable}
Zengyou He and Weichuan Yu.
\newblock Stable feature selection for biomarker discovery.
\newblock \emph{CoRR}, abs/1001.0887, 2010.

\bibitem[Ideker et~al.(2002)Ideker, Ozier, Schwikowski, and
  Siegel]{Ideker2002Discovering}
Trey Ideker, Owen Ozier, Benno Schwikowski, and Andrew~F. Siegel.
\newblock Discovering regulatory and signalling circuits in molecular
  interaction networks.
\newblock In \emph{ISMB}, pages 233--240, 2002.

\bibitem[Jacob et~al.(2009)Jacob, Obozinski, and Vert]{Jacob2009Group}
Laurent Jacob, Guillaume Obozinski, and Jean-Philippe Vert.
\newblock Group lasso with overlap and graph lasso.
\newblock In \emph{ICML '09: Proceedings of the 26th Annual International
  Conference on Machine Learning}, pages 433--440, New York, NY, USA, 2009.
  ACM.
\newblock ISBN 978-1-60558-516-1.
\newblock \doi{http://doi.acm.org/10.1145/1553374.1553431}.

\bibitem[{J}enatton et~al.(2009){J}enatton, {A}udibert, and
  {B}ach]{Jenatton2009Structured}
R.~{J}enatton, J.-Y. {A}udibert, and F.~{B}ach.
\newblock {S}tructured {V}ariable {S}election with {S}parsity-{I}nducing
  {N}orms.
\newblock Research report, {WILLOW} - {INRIA}, 2009.
\newblock URL \url{http://hal.inria.fr/inria-00377732/en/}.

\bibitem[Lehmann and Romano(2005)]{Lehmann2005Testing}
E.~L. Lehmann and Joseph~P. Romano.
\newblock \emph{Testing statistical hypotheses}.
\newblock Springer Texts in Statistics. Springer, New York, third edition,
  2005.
\newblock ISBN 0-387-98864-5.

\bibitem[Loi et~al.(2008)Loi, Haibe-Kains, Desmedt, Wirapati, Lallemand, Tutt,
  Gillet, Ellis, Ryder, Reid, et~al.]{Loi2008Predicting}
S.~Loi, B.~Haibe-Kains, C.~Desmedt, P.~Wirapati, F.~Lallemand, A.M. Tutt,
  C.~Gillet, P.~Ellis, K.~Ryder, J.F. Reid, et~al.
\newblock Predicting prognosis using molecular profiling in estrogen
  receptor-positive breast cancer treated with tamoxifen.
\newblock \emph{BMC Genomics}, 9\penalty0 (1):\penalty0 239, 2008.

\bibitem[L\"{o}nnstedt and Speed(2001)]{Lonnstedt2001Replicated}
Ingrid L\"{o}nnstedt and Terry Speed.
\newblock Replicated microarray data.
\newblock \emph{Statistica Sinica}, 12:\penalty0 31--46, 2001.

\bibitem[Louie et~al.(2010)Louie, McClellan, Siewit, and
  Kawabata]{Louie2010Estrogen}
Maggie~C Louie, Ashley McClellan, Christina Siewit, and Lauren Kawabata.
\newblock Estrogen receptor regulates e2f1 expression to mediate tamoxifen
  resistance.
\newblock \emph{Mol Cancer Res}, 8\penalty0 (3):\penalty0 343--352, Mar 2010.
\newblock \doi{10.1158/1541-7786.MCR-09-0395}.
\newblock URL \url{http://dx.doi.org/10.1158/1541-7786.MCR-09-0395}.

\bibitem[Lu et~al.(2005)Lu, Liu, Xiao, and Deng]{Lu2005Hotelling}
Yan Lu, Peng-Yuan Liu, Peng Xiao, and Hong-Wen Deng.
\newblock Hotelling's t2 multivariate profiling for detecting differential
  expression in microarrays.
\newblock \emph{Bioinformatics}, 21\penalty0 (14):\penalty0 3105--3113, Jul
  2005.
\newblock \doi{10.1093/bioinformatics/bti496}.
\newblock URL \url{http://dx.doi.org/10.1093/bioinformatics/bti496}.

\bibitem[Ma and Kosorok(2009)]{Ma2009Identification}
Shuangge Ma and Michael~R. Kosorok.
\newblock Identification of differential gene pathways with principal component
  analysis.
\newblock \emph{Bioinformatics}, 25\penalty0 (7):\penalty0 882--889, 2009.
\newblock ISSN 1367-4803.
\newblock \doi{http://dx.doi.org/10.1093/bioinformatics/btp085}.

\bibitem[Man(2010)]{Man2010Aberrant}
Yan-Gao Man.
\newblock Aberrant leukocyte infiltration: a direct trigger for breast tumor
  invasion and metastasis.
\newblock \emph{Int J Biol Sci}, 6\penalty0 (2):\penalty0 129--132, 2010.

\bibitem[Musgrove and Sutherland(2009)]{Musgrove2009Biological}
Elizabeth~A Musgrove and Robert~L Sutherland.
\newblock Biological determinants of endocrine resistance in breast cancer.
\newblock \emph{Nat Rev Cancer}, 9\penalty0 (9):\penalty0 631--643, Sep 2009.
\newblock \doi{10.1038/nrc2713}.
\newblock URL \url{http://dx.doi.org/10.1038/nrc2713}.

\bibitem[Nacu et~al.(2007)Nacu, Critchley-Thorne, Lee, and Holmes]{Nacuetal07}
S.~Nacu, R.~Critchley-Thorne, P.~Lee, and S.~Holmes.
\newblock {Gene expression network analysis and applications to immunology}.
\newblock \emph{Bioinformatics}, 23\penalty0 (7):\penalty0 850, 2007.

\bibitem[Rapaport et~al.(2007)Rapaport, Zynoviev, Dutreix, Barillot, and
  Vert]{Rapaport2007Classification}
F.~Rapaport, A.~Zynoviev, M.~Dutreix, E.~Barillot, and J.-P. Vert.
\newblock Classification of microarray data using gene networks.
\newblock \emph{BMC Bioinformatics}, 8:\penalty0 35, 2007.

\bibitem[Stain and Weiss(1971)]{Stain1971Introduction}
Elias~M. Stain and Guido Weiss.
\newblock \emph{Introduction to Fourier Analysis on Euclidean Spaces}.
\newblock Princeton University Press, 1971.

\bibitem[Subramanian et~al.(2005)Subramanian, Tamayo, Mootha, Mukherjee, Ebert,
  Gillette, Paulovich, Pomeroy, Golub, Lander, and
  Mesirov]{Subramanian2005Gene}
A.~Subramanian, P.~Tamayo, V.~K. Mootha, S.~Mukherjee, B.~L. Ebert, M.~A.
  Gillette, A.~Paulovich, S.~L. Pomeroy, T.~R. Golub, E.~S. Lander, and J.~P.
  Mesirov.
\newblock Gene set enrichment analysis: a knowledge-based approach for
  interpreting genome-wide expression profiles.
\newblock \emph{Proc. Natl. Acad. Sci. USA}, 102\penalty0 (43):\penalty0
  15545--15550, Oct 2005.
\newblock \doi{10.1073/pnas.0506580102}.
\newblock URL \url{http://dx.doi.org/10.1073/pnas.0506580102}.

\bibitem[Sutherland and Musgrove(2009)]{Sutherland2009CDK}
Robert~L Sutherland and Elizabeth~A Musgrove.
\newblock Cdk inhibitors as potential breast cancer therapeutics: new evidence
  for enhanced efficacy in er+ disease.
\newblock \emph{Breast Cancer Res}, 11\penalty0 (6):\penalty0 112, 2009.
\newblock \doi{10.1186/bcr2454}.
\newblock URL \url{http://dx.doi.org/10.1186/bcr2454}.

\bibitem[Tai and Speed(2008)]{Tai2008On}
Yu~Chan Tai and Terry Speed.
\newblock On gene ranking using replicated microarray time course data.
\newblock \emph{Biometric}, 65\penalty0 (1):\penalty0 40--51, June 2008.

\bibitem[Vandin et~al.(2010)Vandin, Upfal, and Raphael]{Vandin2010Algorithms}
Fabio Vandin, Eli Upfal, and Benjamin~J. Raphael.
\newblock Algorithms for detecting significantly mutated pathways in cancer.
\newblock In Bonnie Berger, editor, \emph{RECOMB}, volume 6044 of \emph{Lecture
  Notes in Computer Science}, pages 506--521. Springer, 2010.
\newblock ISBN 978-3-642-12682-6.

\bibitem[Vaske et~al.(2010)Vaske, Benz, Sanborn, Earl, Szeto, Zhu, Haussler,
  and Stuart]{Vaske2010Inference}
Charles Vaske, Stephen Benz, Zachary Sanborn, Dent Earl, Christopher Szeto,
  Jingchun Zhu, David Haussler, and Joshua Stuart.
\newblock Inference of patient-specific pathway activities from
  multi-dimensional cancer genomics data using {PARADIGM}.
\newblock In \emph{ISMB}, 2010.

\bibitem[Zhou and Sch{\"o}lkopf(2005)]{Zhou2005Regularization}
Dengyong Zhou and Bernhard Sch{\"o}lkopf.
\newblock Regularization on discrete spaces.
\newblock In Walter~G. Kropatsch, Robert Sablatnig, and Allan Hanbury, editors,
  \emph{DAGM-Symposium}, volume 3663 of \emph{Lecture Notes in Computer
  Science}, pages 361--368. Springer, 2005.
\newblock ISBN 3-540-28703-5.

\end{thebibliography}

\end{document}